\documentclass{amsart}

\usepackage{amsmath}
\usepackage{blkarray}
\usepackage{amssymb}
\usepackage{amsfonts}
\usepackage{amsthm}
\usepackage{bbm}
\usepackage{cancel}
\usepackage{caption}
\usepackage{courier}
\usepackage{circuitikz}
\usepackage{color}
\usepackage{colortbl}
\usepackage{changepage}
\usepackage{csquotes}
\usepackage{enumerate}
\usepackage{forest}
\usepackage{floatrow}
\usepackage{float}
\usepackage{graphicx}
\usepackage{hyperref}
\usepackage{soul}
\usepackage{tabu}
\usepackage{textcomp}
\usepackage{tikz}

\usepackage{url}
\usepackage{verbatim}
\usepackage{wasysym}
\usepackage{yfonts}
\usepackage{changepage}
\usetikzlibrary{automata, arrows, positioning, shapes}

\usepackage[ruled,vlined]{algorithm2e}
\usepackage{multicol}

\usepackage{amsmath,amssymb,amsthm,graphicx,setspace,tabto,fancyhdr,mathtools}
\usepackage{float,tikz}
\usepackage[shortlabels]{enumitem}
\usepackage[nobreak=false,everyline=true]{mdframed}
\usepackage[left=0.75in,right=0.75in,top=1.0in,bottom=2.0in]{geometry}

\usepackage{hyperref}
\hypersetup{
    colorlinks=true,
    linkcolor=blue,
    filecolor=magenta,      
    urlcolor=blue,
}

\usepackage{caption}
\usepackage{subcaption}

\newtheoremstyle{definition}
{3pt} 
{3pt} 
{} 
{} 
{\bfseries} 
{.} 
{.5em} 
{} 

\DeclarePairedDelimiter{\floor}{\lfloor}{\rfloor}

\newtheorem{thm}{Theorem}[section]
\newtheorem{theorem}[thm]{Theorem}
\newtheorem{defn}[thm]{Definition}
\newtheorem{prop}[thm]{Proposition}
\newtheorem{proposition}[thm]{Proposition}
\newtheorem{lemma}[thm]{Lemma}
\newtheorem{cor}[thm]{Corollary}

\newtheorem{example}[thm]{Example}

\newtheorem{remark}[thm]{Remark}

\def\defeq{\mathrel{\mathop:}=}

\newcommand{\w}{\mathbf{w} }
\renewcommand{\b}{\mathbf{b} }
\newcommand{\1}{\mathbf{1}}

\thanks{}

\title{Local Statistics and Shuffling for Dimers on a Square-Hexagon Lattice}

\date{\today}

\author{Matthew Nicoletti}

\begin{document}

\begin{abstract}
    We study the dimer model on special subgraphs of the square hexagon lattice called ``tower graphs" of size $N$. Using integrable probability techniques, we confirm that as $N \rightarrow \infty$, the local statistics are translation invariant Gibbs measures, as conjectured by Kenyon-Okounkov-Sheffield \cite{KOS2006}. We also present a 2+1-dimensional discrete time growth process, whose time $N$ distribution is exactly the dimer model on the size $N$ tower, and we compute the current of this growth process and confirm that the model belongs to the Anisotropic KPZ universality class.
\end{abstract}

\maketitle

\section{Introduction}

\subsection{Background}

A \emph{perfect matching} or \emph{dimer cover} of a graph $G$ is a subset $M$ of edges for which each vertex is incident to exactly one of the edges of $M$. The dimer model first appeared in the 1930s as a model for liquid mixtures with molecules of two very distinct sizes, and the partition function of the model (number of dimer covers) was estimated. Since then, physicists and mathematicians have extensively studied questions about the dimer model from many different perspectives; see the surveys \cite{gorin2021lectures, Kenyon2007Lecture}, as well as \cite{cohn-elki-prop-96, Kenyon2000Conformal, CohnKenyonPropp2000, Kenyon2001GFF, okounkov2003correlation, KOS2006, RVPA, novak2015lozenge, Cimasoni_2007, GoncharovKenyon2011DimersClusterIntSys, Betea_etal2014, laslier2013lozenge, borodin2007periodic} and references therein for just a few examples.

 Many of the central questions are probabilistic in nature. A well studied setup is the following: suppose we have a $\mathbb{Z}^2$-periodic lattice $\mathfrak{L}$, and a bounded domain $R \subset \mathbb{R}^2$. Corresponding to any dimer configuration of a bipartite graph, there is a naturally corresponding \emph{height function}. The main object of study is the random (normalized) height function $h_\epsilon$ arising from a random dimer configuration on a subset of the lattice $\epsilon \mathfrak{L}$ approximating $R$. The approximating sequence of finite domains is chosen so that the height function satisfies a given boundary condition on $\partial R$.
 
 It is known that under quite general circumstances, as $\epsilon \rightarrow 0$, the functions $h_\epsilon$ converge to a deterministic limiting function $\mathfrak{h}  : R \rightarrow \mathbb{R}$, which is the solution of a certain variational principle \cite{CohnKenyonPropp2000, KOS2006}. Furthermore, there is a very general and precise conjecture of Kenyon-Okounkov \cite{OkounkovKenyon2007Limit} predicting the convergence in distribution of the fluctuation field $h_\epsilon - \mathfrak{h}$ to a certain Gaussian Free Field, and there have been proofs of this conjecture in many special cases (e.g. in \cite{Petrov2012GFF, Ferrari2008, BufetovKnizel}), as well as for a very general class of boundary conditions on the hexagonal lattice \cite{huanglozengegff}. We note also that the GFF fluctuations have been proven for a large class of boundary conditions (including ours) on the lattice we study in this paper \cite{BoutillierLiSHL}.

Another central question involves the study of \emph{local statistics}: What describes the statistics of the random dimer configuration in a finite neighborhood (at lattice scale) of a given interior point $(x, y) \in R$? It is conjectured that the answer is given by a certain \emph{ergodic, translation invariant Gibbs measure} on dimer covers of $\mathcal{L}$, corresponding to the slope of the limiting height function at $(x, y)$. This again has been verified in several special cases \cite{cohn-elki-prop-96, Petrov2012, Ferrari2008}, and again only recently for general boundary conditions for the hexagonal lattice \cite{aggarwal2019universality}.

Other questions surrounding the dimer model involve studying its connections to various 1+1 and 2+1-dimensional random growth processes in the KPZ universality class (see the introduction of \cite{BorFerr2008DF}, and references therein). As a concrete example, this connection arises from the shuffling algorithm for domino tilings (see \cite{propp2003generalized} for background on domino shuffling), which can also be viewed as a Markov chain on interlacing arrays of particles \cite{Nordenstam_Aztec_2009, borodin2015random}. In the domino tiling case, the top portion of the arctic curve can itself can be viewed as a fixed time slice of a 1+1-dimensional growth process obtained as a Markovian projection of the interlacing particle dynamics. This fact was utilized in the first proof of the so-called ``Arctic Circle Theorem" \cite{Jockusch1998RandomDT}. Furthermore, the 2+1-dimensional growth process coming from the height function evolution under the shuffling algorithm is a member of the KPZ universality class. The speed of growth of the height function as a function of local slopes $(s, t)$, also known as the \emph{current}, has been explicitly computed, and various conjectures about critical exponents describing the size of height fluctuations have been verified \cite{chhita2019, Chhita_2021}. Furthermore, starting from an arbitrary initial condition the height function has been shown to converge in the hydrodynamic limit to the solution of an explicit PDE \cite{ZhangDominoHydrodynamics}.

In this note, we focus on local statistics and on an associated 2+1-dimensional growth process, which is the analog of domino shuffling for our graph. The finite graphs we study, called \emph{tower graphs}, which are subgraphs of a certain $\mathbb{Z}^2$-periodic lattice, are a special case of the rail-yard graphs from \cite{boutillier2015dimers}. See Figure \ref{fig:tower} for a size $N = 2$ tower graph.

We are able to identify a sub-matrix of a finite dimensional kernel, which describes an associated determinantal measure on interlacing particle arrays, with the inverse Kasteleyn matrix of the dimer model on the size $N$ tower graph. Then, we take asymptotics as $N \rightarrow \infty$ and match the limiting determinantal processes in the bulk with the translation invariant Gibbs measures (c.f. \cite{KOS2006}). The asymptotic analysis proceeds via a steepest descent analysis, and one of the essential pieces in describing the limit is the critical point of the ``action function". We identify this critical point with the solution of the complex Burgers equation (this is the PDE which is well known to describe the limit shape  \cite{OkounkovKenyon2007Limit}).

In terms of the relationship to growth processes, our story is analogous to the Aztec diamond case; there is an interacting particle process, first described in Section 4.2 of \cite{borodin2015random}, which can be used to sample a random matching of a size $N$ tower graph. We describe how this Markov chain can be viewed as a shuffling algorithm arising from applying a particular sequence of ``urban renewal" moves (see Section \ref{subsec:shuffling}, and also see \cite{propp2003generalized} and \cite[Section 2]{ZhangDominoHydrodynamics} for details on urban renewal moves and their interpretation as a random mapping). We also explicitly compute the speed of growth, or the \emph{current}, as a function of the local slopes. We use this to argue that this 2+1-dimensional growth model is a member of the Anisotropic KPZ universality class. We also conjecture a hydrodynamic limit equation for this growth process for an arbitrary initial condition, and we verify this conjecture via explicit calculations in the case of uniform weights for the initial condition corresponding to tower graphs.

\begin{figure}
\includegraphics[scale=0.8]{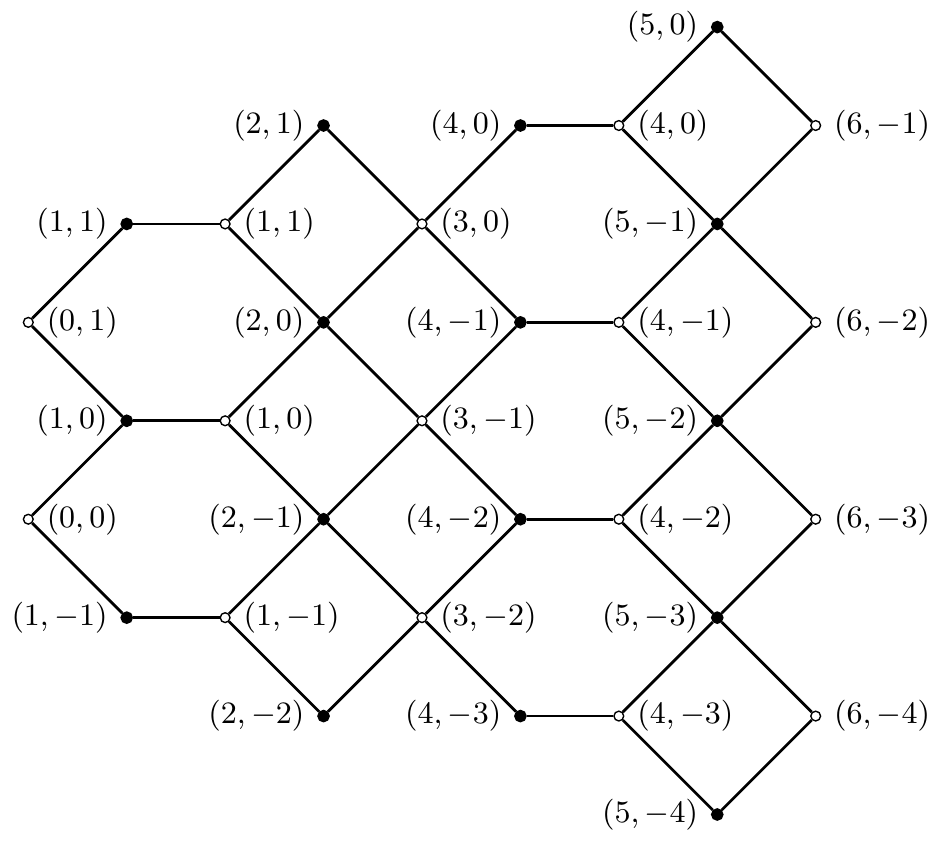}
\caption{The size $N = 2$ tower graph. Vertices are labelled with their $(X, U)$ coordinates described in Section \ref{sec:schur_p}.}
\label{fig:tower}
\end{figure}

\subsection{Main Results and Plan}
\label{subsec:res}
We study random perfect matchings on the size $N$ \emph{tower graph} described in \cite{borodin2015random}, Section 4.2. Thus, in our case the lattice $\mathcal{L}$ is the \emph{square-hexagon lattice} (see Figure \ref{fig:SHL} for the lattice, and see Figure \ref{fig:mag} for the fundamental domain), and the tower of size $N$ is a certain finite subgraph of $\mathcal{L}$ (see Figure \ref{fig:tower} for $N = 2$). The probability measures $P_N$ which we study have two parameters, $\alpha, \beta > 0$, which are weights on the edges of the graph.

As elucidated in Section \ref{subsec:schur}, we may equivalently view a perfect matching as a set of particles placed on lattice sites of a certain finite subset $ \mathfrak{T}_N \subset \mathbb{Z}^2$; this perspective will be useful in describing our results below. See Figure \ref{fig:bijection} for an example of this correspondence between perfect matchings and interlacing arrays. Furthermore, the probability measure we study can be seen as the time $N$ distribution of an interacting particle system: Given a configuration corresponding to a random matching of a size $N$ tower, one can obtain a random configuration corresponding to a matching of a size $N+1$ tower as follows. A configuration can be seen a list of integer arrays
$$T(N) = y^1, x^1, z^2, y^2, x^2, \dots, y^N, x^N, z^{N+1}, y^{N+1}, x^{N+1} $$
where $z^{N+1} = \{-1, -2, \dots, -2 N\}, y^{N+1} =\{-1, -2, \dots, -2 N,-2 N - 1\}, x^{N+1} = \{-1, -2, \dots, -2 N-2\}$ are deterministically ``empty" (in the sense that they correspond to the empty Maya diagram).

\begin{figure}
\includegraphics[scale=0.8]{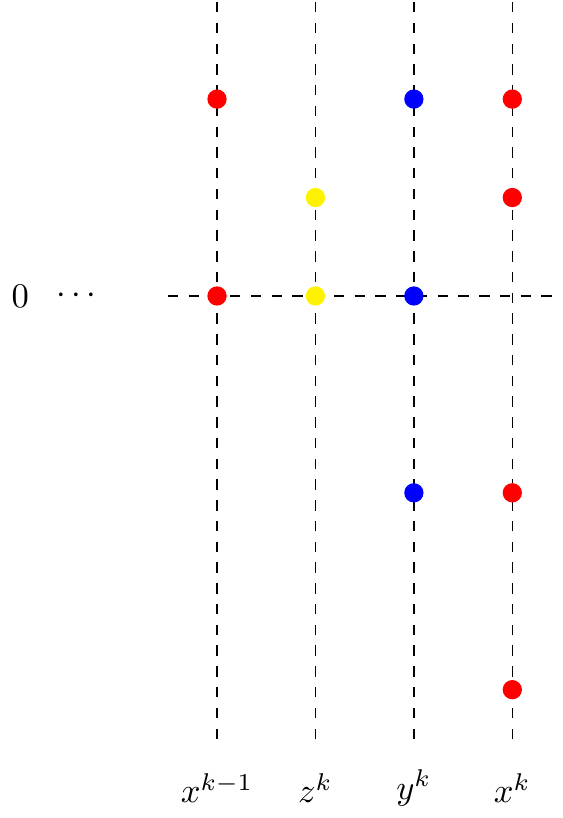}
\caption{An example of the particle arrays satisfying the required interlacing conditions.}
\label{fig:ptcl}
\end{figure}

\begin{figure}[htb]
    \centering 
\begin{subfigure}{0.25\textwidth}
  \includegraphics[width=\linewidth]{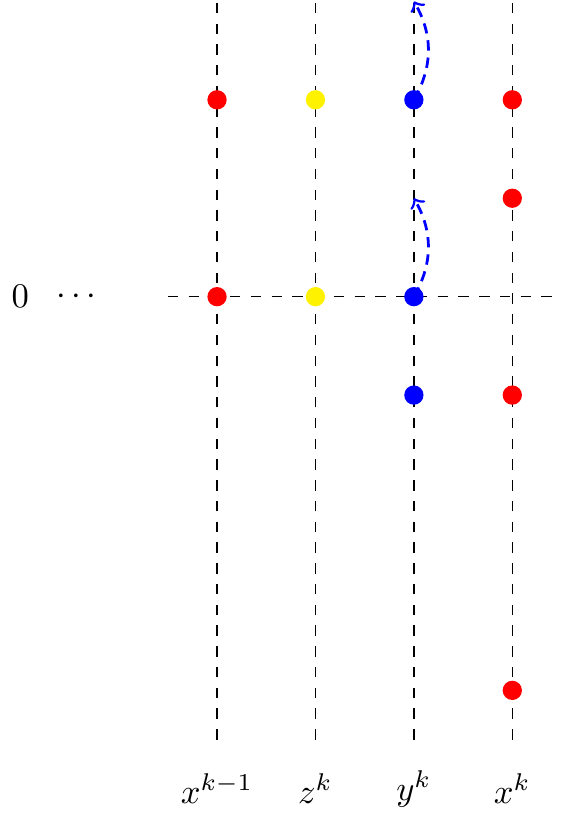}
  \caption{After the deterministic update of $z^k$ but before step 1.}
  \label{fig:1}
\end{subfigure}\hfil 
\begin{subfigure}{0.25\textwidth}
  \includegraphics[width=\linewidth]{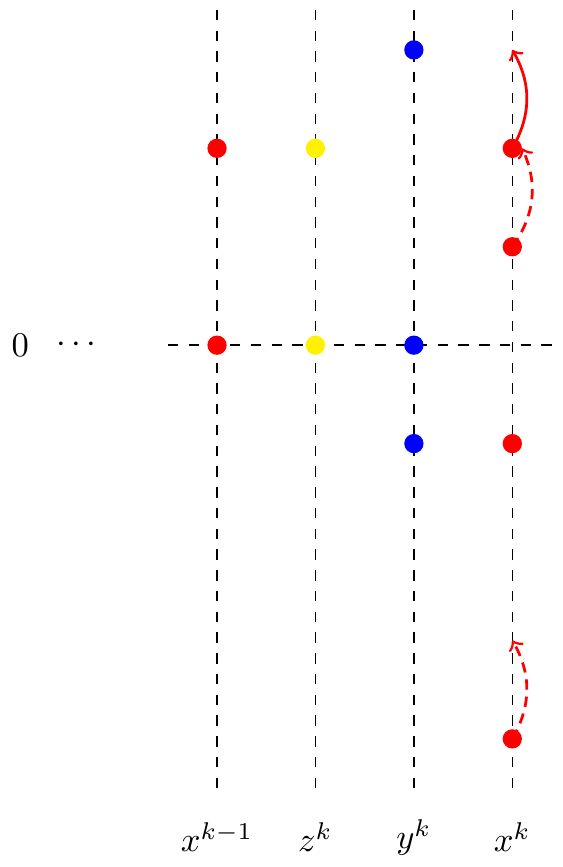}
  \caption{After step 1.}
  \label{fig:2}
\end{subfigure}\hfil 
\begin{subfigure}{0.25\textwidth}
  \includegraphics[width=\linewidth]{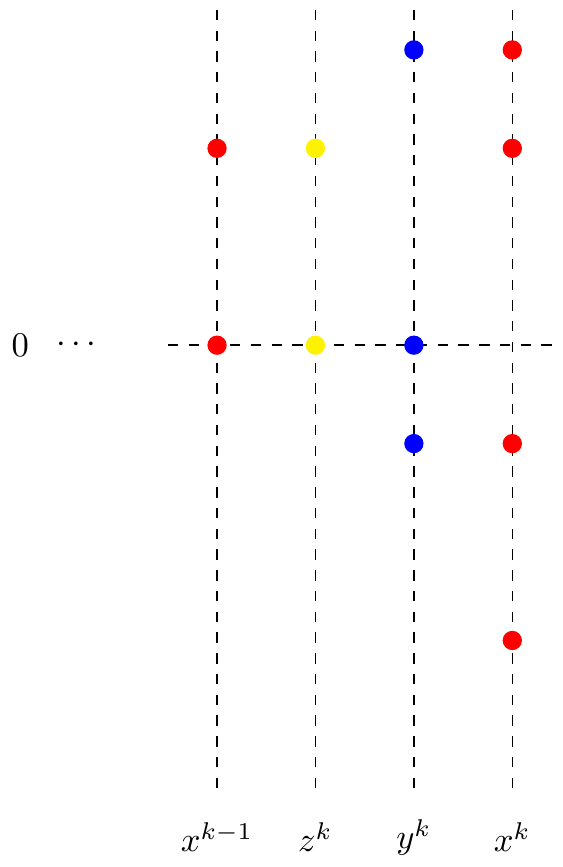}
  \caption{After step 2.}
  \label{fig:3}
\end{subfigure}
\caption{Shown above is an example of the updates of $z^k, y^k, x^k$ during one step of the Markov chain. The dashed arrows denote possible jumps, the solid arrows denote forced jumps, and particles with no arrows are blocked.}
\end{figure}

Then perform the following steps. First, update $z^i$ to $z^i(N+1) = x^{i-1}$ deterministically. Then do the two rounds of updates (these two steps precisely correspond to the two steps in the shuffling formulation in Section \ref{subsec:shuffling}):
\begin{enumerate}[1.]
\item Given the updated $\{z^i_j\}$, update each particle position in $\{y^i_j\}$ independently. The update $y^i_j \rightarrow y^i_j(N+1)$ happens deterministically if forced or blocked:
\begin{itemize}
\item If $x_j^{i-1}(N) = z_j^i(N+1) = y^i_j + 1$, then set $y^i_j(N+1) = y^i_j + 1$. In this case we say a \emph{jump is forced} to preserve interlacing.
\item If $x_{j-1}^{i-1}(N) = z_{j-1}^i(N+1) = y^i_j + 1$, then set $y^i_j(N+1) = y^i_j$. In this case we say $y^i_j$ is \emph{blocked} to preserve interlacing.
\end{itemize}
Otherwise $y^i_j(N+1) = y^i_j + 1$ with probability $\frac{\alpha \beta}{1+\alpha \beta}$, and stays otherwise.
\item Given the updated $\{y^i_j\}$, update each particle position in $\{x^i_j\}$ independently. Update $x^i_j$ deterministically if forced or blocked by the new particles at $\{y^i_j(N+1)\}$ to preserve interlacing, as described in step 1. Otherwise jump by $1$ with probability $\frac{\beta}{1+\beta}$, or stay otherwise.
\end{enumerate}
The resulting particle arrays 
$$T(N+1) = y^1(N+1), x^1(N+1), \dots, y^N(N+1), x^N(N+1), z^{N+1}(N+1), y^{N+1}(N+1), x^{N+1}(N+1)$$
correspond to a matching of a tower of size $N+1$. In this way, a random configuration $T(N)$ can be built up from the empty one at time $N = 0$ using the Markov chain above.

\subsubsection{Determinantal Kernel}

In Section \ref{sec:schur_p} we show that the random particle configuration at time $N$ is a \emph{determinantal point process} and we compute its kernel by identifying the distribution with a \emph{Schur process}. Furthermore, we identify a certain restriction of the kernel with the inverse of the Kasteleyn matrix of the size $N$ tower graph, giving an explicit formula for the inverse Kasteleyn matrix.

\subsubsection{Local Statistics}
In Section \ref{sec:loc} we compute the local statistics of a random matching of the tower graph of size $N$ near an arbitrary macroscopic point of the domain, and identify the limit with the ergodic translation invariant Gibbs measure of the correct slope. We describe the result more precisely below.

Corresponding to each size $N$ tower perfect matching is a \emph{height function} 
$$H_N :\mathfrak{T}_N \rightarrow \mathbb{Z} $$
where $\mathfrak{T}_N$ is a certain finite subset of $\mathbb{Z}^2$. In particular, in each column there are only finitely many particles, and we define for $(X, U) \in \mathfrak{T}_N$
\begin{align*}
H_N(X, U) = - \# \{\text{particles at } (X, V) : V > U \} \;\;.
\end{align*}
If we choose a particular reference matching for the dimer model, $H(X, U)$ coincides with the value of the height function defined by the corresponding perfect matching at the face just above (with respect to the embedding of Figure \ref{fig:SHL}) the vertex or pair of vertices with coordinates $(X, U)$.

It is known (see \cite{KOS2006, BoutillierLiSHL}) that $\frac{1}{N} H_{N}(3 N x, N u)$ converges in probability to a deterministic limit $h(x, u)$, where $h$ is defined on the domain $\mathfrak{T} \coloneqq \{(x, u) \in \mathbb{R}^2 : 0 \leq x \leq 1 , - 2 x \leq u \leq 1- x \}$ (this set is approximated by $(X/(3 N), U /N)$ for $(X, U) \in \mathfrak{T}_N$ as $N$ grows). The \emph{liquid region} $\mathfrak{L} \subset \mathfrak{T}$ is the region where the graph of $h$ is curved, and the \emph{frozen region} is the region where $h$ is a linear function. The \emph{arctic curve} $C = \partial \mathfrak{L}$ separating them is an algebraic curve. For each pair of allowed height function slopes $(s, t)$ (see Section \ref{sec:loc} for the precise definition of the slopes $(s, t)$ of a height function) there is an \emph{ergodic translation invariant Gibbs measure} $\pi_{s, t}$ on dimer covers of the lattice $\mathcal{L}$. The following theorem states that local statistics are given by $\pi_{s, t}$ in the limit. 

\begin{thm}
Let $(x, u) \in \mathfrak{T}\setminus C$ denote a rescaled position in a large tower graph, which is away from the arctic curve $C$. The dimer configuration in any finite neighborhood of the lattice site $(\floor{3 N x}, \floor{N u})$ converges in distribution to that of $\pi_{s, t}$, where $(s, t)$ are the slopes of the limiting height function $h$ at $(x, u)$.
\label{thm:mainthm}
\end{thm}

Theorem \ref{thm:mainthm} follows from Theorem \ref{thm:loc} in the text (see Corollary \ref{cor:main}). In particular, we are able to compute the height function slopes at any point $(x, u)$ in the domain via the critical point $z(x, u)$ of the action, see subsection \ref{subsec:BL}. In Section \ref{sec:burgers}, we identify the critical point $z(x, u)$ with the complex coordinate $z$ which is the solution of the \emph{complex Burgers equation}, and we describe the mapping between the complex coordinate and slopes $(s, t)$ (see subsection \ref{subsec:z}). 

For general $\alpha$, $z(x, u)$ is the solution of a cubic equation, but when $\alpha = 1$, the cubic factors and the relevant solution $z(x, u)$ is the root of an explicit quadratic, and thus can be solved for explicitly. In Example \ref{ex:uni} we give the explicit formula for the height function in the uniform case.

\subsubsection{Surface Growth}

In Section \ref{sec:growth}, we prove several properties of the growth process under which $H_N$ evolves. We describe a shuffling algorithm (a sequence of randomized urban renewal moves) which leads to the same dynamics as the particle system described above, and also define a full plane version of the same Markov process, which acts on height functions $H : \mathbb{Z}^2 \rightarrow \mathbb{Z}$ coming from dimer covers of the full lattice $\mathcal{L}$, and we show that the Gibbs measures $\{\pi_{s, t}\}$ are stationary under the process. We compute the speed of growth, or \emph{current}, 
$$J(s, t) \coloneqq \mathbb{E}_{\pi_{s, t}} [H_1(0,0) - H_0(0, 0)]$$
of the growth process, where $\mathbb{E}_{\pi_{s, t}}$ denotes the expectation over one time step of the process with respect to a height function $H_0$ initialized at a Gibbs measure $\pi_{s, t}$. We find that when written in terms of the complex magnetic field coordinate $z$ in the upper half plane $\mathbb{H}$ (this coordinate is in one-to-one correspondence with slopes $(s, t)$, see subsections \ref{subsec:gibbs_meas} and \ref{subsec:z}), $J(z)$ is a harmonic function of $z$. More precisely, we have

\begin{thm}
When written in terms of the complex coordinate $z \in \mathbb{H}$, the current is 
$$J(z)  = -\frac{1}{\pi} \arg(1 + \beta z) $$
where $\arg$ takes values in $(-\pi, \pi]$. 
\end{thm}

By a result of Borodin-Toninelli \cite{borodin2018two}, this allows to identify the growth process as a member of the Anisotropic KPZ universality class. Furthermore, we use our calculation of the current $J$ to predict a hydrodynamic limit equation for the time evolution of 
$h_{\tau}(x, u) = \lim_{N \rightarrow \infty} \frac{1}{N} H_{N \tau}(N x, N u) $, and we verify this hydrodynamic limit equation explicitly in the uniform case.

\subsection{Acknowledgements}
The author happily thanks Alexei Borodin for guidance and many useful discussions, as well as David Keating and Tomas Berggren for several helpful comments and suggestions.

\section{Inverse Kasteleyn and Schur Process}
\label{sec:schur_p}

\subsection{Dimer Model}

\begin{figure}
\includegraphics[scale=0.8]{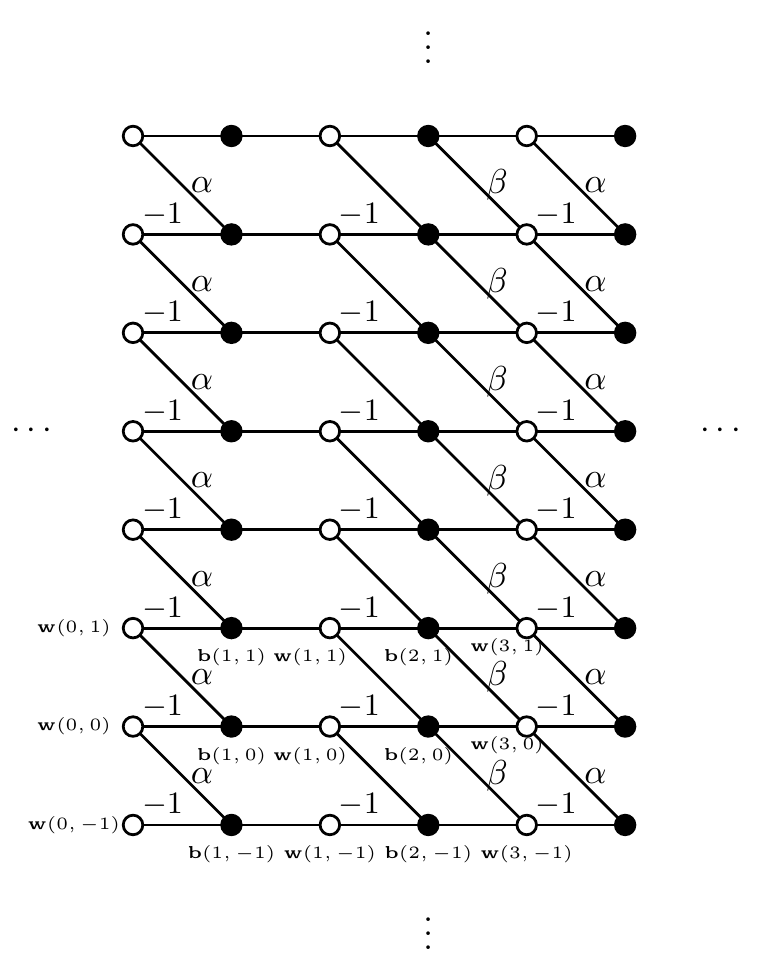}
\caption{Above is the square hexagon lattice $G$, with edge weights and Kasteleyn signs shown.}
\label{fig:SHL}
\end{figure}

\begin{figure}
\includegraphics[scale=0.8]{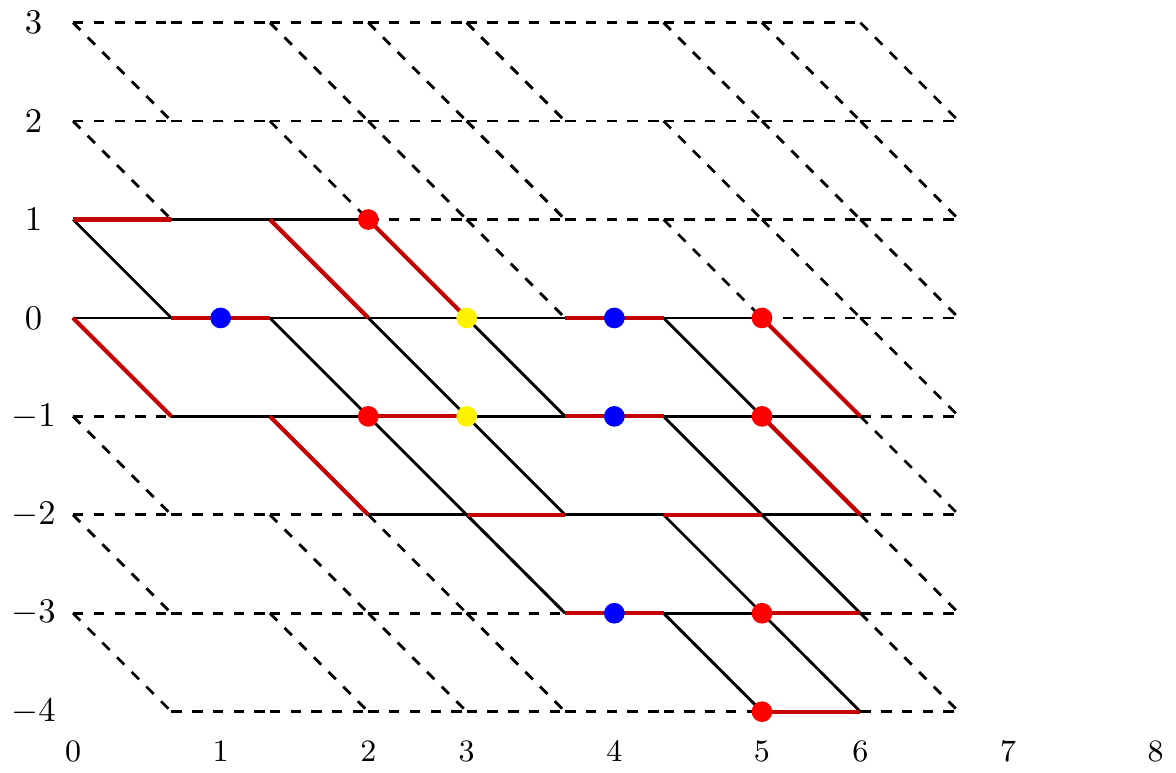}
\caption{The correspondence between a matching of the tower graph and a collection of interlacing particles for $N = 2$.}
\label{fig:bijection}
\end{figure}

We now define the model we study. We will consider finite subgraphs of the graph $\mathcal{L}$ in Figure \ref{fig:SHL}, which we refer to as the \emph{square hexagon lattice}. The set of vertices are indexed by pairs $(X, U) \in \mathbb{Z}^2$. We have
\begin{itemize}
\item a white vertex $\w(X, U)$ if $X = 0\; (\text{mod } 3)$ or $X = 1 \; (\text{mod } 3)$ 
\item a black vertex $\b(X, U)$ if $X = 2 \; (\text{mod } 3)$ or $X = 1 \; (\text{mod } 3)$ 
\end{itemize}
and we call the white and black vertex sets $V_W, V_B$, respectively. Note that our choice of coordinates is such that if $X = 0,2 \;(\text{mod } 3)$ there is either a white or a black vertex corresponding to this pair, and if $X = 1 \;(\text{mod } 3)$ we have a white vertex $\w(X, U)$ and a black vertex $\b(X, U)$. For $\w(X, U)$ a white vertex, we have the edges
\begin{itemize}
\item $\w(X, U) - \b(X, U) $ if $X = 1 \; (\text{mod } 3)$ 
\item $\w(X, U) - \b(X-1, U) $ if $X = 0 \; (\text{mod } 3)$
\item $\w(X, U) - \b(X+1, U) $ 
\item $\w(X, U) - \b(X-1, U+1)$ if $X = 0 \; (\text{mod } 3)$
\item $\w(X, U) - \b(X+1, U-1) $ with weight $\alpha$ if $X = 0 \; (\text{mod } 3)$
\item $\w(X, U) - \b(X+1, U-1) $ with weight $1$ if $X = 1 \; (\text{mod } 3)$
\end{itemize}
Using this embedding, we define coordinates on the faces of the lattice as follows: If $(X,U) \in \mathbb{Z}^2$ are the coordinates of a (black or white) vertex in the square-hexagon lattice, let $(X, U)$ be the coordinates of the face $F$ containing the point $(X, U + \frac{1}{2}) \in \mathbb{R}^2$.

\begin{defn}[Tower graph of size $N$]
Let the \emph{tower of size $N$} be the subgraph induced by black and white vertices $\w(X, U), \b(X, U)$ at $(X, U)$ such that $0 \leq X \leq 3 N $ and $-X+ \floor{\frac{X}{3}}  \leq U \leq N- \floor{\frac{X}{3}}-1$. See Figure \ref{fig:bijection}.
\end{defn}

We study the probability measure on \emph{perfect matchings}, which are subsets $M$ of edges such that each vertex in the graph is incident to exactly one of the edges, where the probability of a matching $M$ is 
\begin{equation}
\label{eqn:btz}
\frac{1}{Z} \prod_{e \in M} \text{wt}(e)  \;\;.
\end{equation}
Above $Z \defeq \sum_{M} \prod_{e \in M} \text{wt}(e) $ is the \emph{partition function}. We also refer to a perfect matching $M$ as a \emph{dimer cover}.

After a choice of an appropriate \emph{reference matching} $M_0$ of the square hexagon lattice, both in the case of the (finite) tower graph and in the case of the full lattice $\mathcal{L}$, there is a one to one mapping from dimer covers to \emph{height functions}, which are real valued functions defined on the set of faces of (a finite subset of) $\mathcal{L}$. See Section 2.2 of \cite{Kenyon2007Lecture} for the general construction of dimer model height functions. In our case we simply choose $M_0$ to be the set of edges labelled with a $-1$ in Figure \ref{fig:SHL}. Then, we make the choice to set $H(0, 0) = 0$ (note that in the finite case, $(0,0)$ is a boundary face, as it is not a face of the tower graph, but it is adjacent to faces of the tower). Given a perfect matching $M$ of the graph $G$ in consideration, in order to determine the height at a face $(X, U)$, one constructs a face path from $(0, 0)$ to $(X, U)$ in $G$, and the height change when crossing from $F \rightarrow F'$ is given by 
\begin{align*}
H(F') - H(F) =
\pm (\mathbf{1}_{e \in M} - \mathbf{1}_{e \in M_0} )
\end{align*}
where the sign is $+$ if we cross the edge with the white vertex on the right, and $-$ otherwise. This is well defined because the height change when making a loop around a vertex is easily seen to be $0$. Furthermore, it can also be checked that this definition of $H$ coincides with the one given in the introduction.

The \emph{Kasteleyn matrix} $\mathbf{K}$ of a graph has rows indexed by black vertices and columns indexed by white vertices, and its entries are defined by 
\begin{align*}
\mathbf{K}(\mathbf{b}, \mathbf{w}) 
=
\begin{cases}
\text{sign}(\mathbf{b}, \mathbf{w}) \cdot \text{wt}(\mathbf{b}, \mathbf{w}) & \text{ if $(\mathbf{b}, \mathbf{w})$ is an edge} \\
0 & \text{ otherwise}
\end{cases}
\end{align*}
where the sign of each edge is chosen so that for each face, the product of signs of its boundary edges is given by $(-1)^{k+1}$, if the face has $2 k$ edges. Such a choice of signs on edges is called a \emph{Kasteleyn weighting} of the graph. The signs of edges, as well as the actual edge weights, for the square hexagon graph $G$ and its finite subgraphs are shown in Figure \ref{fig:SHL}. We will use $\mathbf{K}$ to denote the Kasteleyn matrix for the tower graph of size $N$, as well as for the square hexagon lattice $G$, and which one is meant should be clear from the context.

It is well known that on any finite graph with a Kasteleyn weighting and Kasteleyn matrix $\mathbf{K}$, we have 
\begin{theorem}\cite{Kasteleyn1967,Kenyon2007Lecture}
The dimer model partition function is given by
   $$Z = | \det \mathbf{K} |.$$
\end{theorem}
As a corollary, if we know the inverse Kasteleyn matrix, we can also compute the correlation functions.
  \begin{cor}\cite{Kenyon2007Lecture}
 Given a set of edges $X = ((\mathbf{w}_1, \mathbf{b}_1), (\mathbf{w}_2, \mathbf{b}_2), \dots, (\mathbf{w}_k, \mathbf{b}_k))$, the probability that all of the edges in $X$ occur in a dimer cover is
  	 	$$(\prod_{i = 1}^k \mathbf{K}(\mathbf{b}_i, \mathbf{w}_i)) \det (\mathbf{K}^{-1} (\mathbf{w}_i, \mathbf{b}_j) )_{1 \leq i, j \leq k} \;\;.$$
\end{cor}
Thus, the set of edges of a random perfect matching sampled from \eqref{eqn:btz} form a determinantal point process.

\subsection{Schur Process and Interlacing Particle Process} 
\label{subsec:schur}
We review a measure preserving bijection between perfect matchings of the size $N$ tower and certain sequences of interlacing particle arrays given in \cite{borodin2015random}, as well as a Markov chain on the particle arrays which allows perfect sampling from the dimer model on the size $N$ tower. An analogous bijection exists for domino tilings of the Aztec diamond, plane partitions, and in many other situations. See \cite{BorodinGorinSPB12, gorin2021lectures} and references therein for more details.

Suppose we have a perfect matching of the tower of size $N$ made up of the collection of edges $M$. Then we define the sequence of particle arrays $y^1, x^1, z^2, y^2, x^2, \dots, z^N y^{N}, x^{N}$ as follows:
\begin{align*}
z^t &\defeq \{U : (\mathbf{b}(3t - 3-1, U-1),\mathbf{w}(3t - 3, U)) \in M \text{ or } (\mathbf{b}(3t - 3-1, U),\mathbf{w}(3t - 3, U)) \in M\} \\
y^t &\defeq \{U : (\mathbf{b}(3t - 2, U),\mathbf{w}(3t - 2, U)) \in M \}  \\
x^t &\defeq \{U : (\mathbf{b}(3t - 1, U),\mathbf{w}(3t , U)) \in M \text{ or } (\mathbf{b}(3t - 1, U),\mathbf{w}(3t , U-1)) \in M\} \;\;.
\end{align*}

Index the particles in $z^k$, $y^k$, and $x^k$ are $(z^k_1 > z^k_2 > \cdots > y^k_{2 k -2})$, $(y^k_1 > y^k_2 > \cdots > y^k_{2 k -1})$, and $(x^k_1 > x^k_2 > \cdots > x^k_{2 k})$, respectively. One may immediately see that for each $k = 1,\dots, N$
\begin{align*}
-(2 k - 2) &\leq z^k_i \leq N - k \\
-(2 k - 1) &\leq y^k_i \leq N - k \\
-2 k  &\leq x^k_i \leq N - k 
\end{align*}
and also we have the interlacing conditions
\begin{align*}
x^k_i &\geq y^k_i > x^k_{i+1} \;\;\; 1 \leq i \leq 2 k - 1 \\
x^{k-1}_i &\geq z^k_i \geq x^{k-1}_i-1 \;\;\; 1 \leq i \leq 2 k - 2 \;\;.
\end{align*}
which we denote by $y^k \prec x^k$, and $x^{k-1} \succ' z^k$, respectively. (Above we have set $z^{N+1} = (-1,-2,\dots, -2 N)$).

This map is in fact a bijection: Given any sequence of particle arrays 
$$T(N) = \emptyset \prec y^1 \prec x^1 \succ' z^2 \prec y^2 \prec x^2 \succ' \cdots \prec y^N \prec x^N \succ' \emptyset$$
one can uniquely reconstruct the corresponding matching $M$. See Figure \ref{fig:bijection} for an illustration, and see \cite{borodin2015random} for a more detailed discussion.

In order to understand the corresponding probability measure on interlacing arrays, we first define the partitions $\lambda^{(k)} = x^k + \delta^{(2 k)}$, $\mu^{(k)} = y^k + \delta^{(2 k-1)}$, $\nu^{(k)} = z^k + \delta^{(2 k-2)}$, where $\delta^{(\ell)} = (1,2,\dots, \ell)$. For partitions (which we think of as having infinitely many zeros appended to the end), interlacing means
\begin{align*}
\mu &\prec \lambda \iff \lambda_1 \geq \mu_1 \geq \lambda_2 \geq \cdots \\
\mu &\succ' \lambda \iff \lambda_i \leq \mu_i \leq \lambda_i+1\;\;\; \forall \; i \geq 1 \;\;.
\end{align*}

 After extending each array of particles to infinity by adding particles at every position less than some negative integer $m$, such that the number of particles at positions $\geq 0$ are the same as the number of holes at positions $\leq -1$, the arrays of particles are known as the \emph{Maya diagrams} of the corresponding partitions. Then we have

\begin{align}\label{eqn:SP}
&\text{Prob}(T(N) =y^1, x^1, z^2, y^2, x^2, \dots, z^N, y^N, x^N ) \notag \\
&= \frac{1}{Z} s_{\lambda^{(N)}}(\rho^-) s_{\lambda^{(N)}/\mu^{(N)}}(\rho^+_1) s_{\mu^{(N)} / \nu^{(N)}}(\rho^+_2) s_{\lambda^{(N-1)} /  \nu^{(N)}}(\rho^-) s_{\lambda^{(N-1)}/\mu^{(N-1)}}(\rho^+_1) s_{\lambda^{(N-1)}/\mu^{(N-1)}}(\rho^+_2)  \\
& \cdots  s_{\lambda^{(1)}/\nu^{(2)}}(\rho^-) s_{\lambda^{(1)}/\mu^{(1)}}(\rho^+_1) s_{\mu^{(1)}}(\rho^+_2) \notag \;\;.
\end{align}
Above, the notation $s_{\lambda/\mu}(\rho)$ refers to the image of the Schur function $s_{\lambda/\mu}$ under the specialization $\rho : \Lambda \rightarrow \mathbb{C}$ of the algebra $\Lambda$ of symmetric functions. We denote by $\rho^-$ the dual specialization with single nonzero variable $\beta$, and $\rho^+_1, \rho^+_2$ refer to the single variable specializations with nonzero variable $1, \alpha$, respectively. Concretely, we have 
\begin{align*}
s_{\mu/\nu}(\rho^+_2) &= \alpha^{|\mu| - |\nu|} \mathbf{1}_{\nu \prec \mu}\\
s_{\lambda/\mu}(\rho^+_1) &= \mathbf{1}_{\mu \prec \lambda}
\end{align*}
and
$$s_{\nu/\lambda}(\rho^-) = \beta^{|\nu| - |\lambda|} \mathbf{1}_{\nu \succ' \lambda} \;\;.$$
When interlacing arrays are equipped with this probability measure, the bijection between dimer covers and interlacing arrays is a measure preserving bijection. The probability measure in equation \eqref{eqn:SP} is a special case of a \emph{Schur process}.

\subsection{Space-time Correlation Kernel}

Define 
$$\mathfrak{T}_N \defeq \{(X, U) \in \mathbb{Z}^2 : 1 \leq X \leq 3 N-1 , -X+ \floor{\frac{X}{3}}  \leq U \leq N- \floor{\frac{X}{3}}-1\} \;\;.$$ Using well known results about Schur processes (see \cite{okounkov2003correlation, borodin2005eynard, aggarwal2015correlation_schur}), if interlacing arrays $y^1, x^1, z^2, y^2, x^2, \dots, z^N, y^N, x^N$ are sampled according to the measure in equation \eqref{eqn:SP}, then the corresponding random set of points $\mathcal{P} \subset 2^{\mathfrak{T}_N}$ consisting of  
\begin{align*}
&\{(3 t - 1, x^t_j) : 1 \leq t \leq N, 1 \leq j \leq 2 t\} \\
&\bigcup \{(3 t - 2, y^t_j): 1 \leq t \leq N, 1 \leq j \leq 2 t-1\}\\
& \bigcup \{(3 t - 3, z^t_j) : 2 \leq t \leq N, 1 \leq j \leq 2 t-2\}
\end{align*}
 is a \emph{determinantal point process}. This means that there is a kernel $(K(P, P'))_{P, P' \in \mathfrak{T}_N}$ satisfying the following property: The probability of $\mathcal{P}$ containing points at locations $P_1,\dots,P_k$ is given by
$$\text{Prob}(P_1,\dots, P_k \in \mathcal{P}) = \det(K(P_i, P_j))_{i,j=1}^k\;\;.$$

For general Schur processes, there are well known explicit contour integral formulas for $K$. We specialize the formula in Theorem 2.2 of \cite{borodin2005eynard}, and we obtain that the correlation kernel for the Schur process corresponding to a size $N$ tower is given as follows. Let $X_i = 3 k_i - m_i$ with $1 \leq m_i \leq 3$ for $i =1, 2$, and let 
$$\Phi( 3 t - m, N; z) \defeq  \frac{(1+\beta z)^{N-t+1}}{(1-z^{-1})^{-(t-1 ) - \mathbf{1}_{m =1}} (1-\alpha z^{-1})^{-(t-1 ) - \mathbf{1}_{m \leq 2}}} $$
for $m \in \{1,2,3\}$. 

Then 
\begin{align}
\label{eqn:kast}
K((X_1, U), (X_2, V)) = \frac{1}{(2 \pi i)^2} \int_{\Gamma_{ \pm}} \int_{\Gamma_{\mp}} \frac{\Phi( 3 k_1 - m_1, N; z)}{\Phi( 3 k_2 - m_2, N; w)}  \frac{w}{z-w} \; \; \frac{dw dz}{z^{U+1} w^{-V+1}} \;\;.
\end{align}

For small $\beta$ and $\alpha$, the contours can be taken as $\Gamma_{+} = \{|z| = 1 + \epsilon\}, \Gamma_- = \{|w| = 1+\epsilon/2 \}$ for small $\epsilon> 0$ if $X_1 \geq X_2$, and we use $\Gamma_-$ for the $z$ integral and $\Gamma_+$ for $w$ if $X_1 < X_2$. If we think of this as the formal sum of residues, which is a rational function in $\alpha$ and $\beta$, then we can analytically continue that formula for $K$ to arbitrary $\beta, \alpha > 0$. Thus, for arbitrary $\beta, \alpha > 0$ we choose contours that contain $0$ and the poles at $1, \alpha$, but not the pole at $-\frac{1}{\beta}$, with $\Gamma_+$ containing $\Gamma_-$ in its interior if $X_1 \geq X_2$, and the other way around if $X_1 < X_2$.

\begin{remark}
\label{rmk:spacetime}
 Label the Maya diagrams as $a^n$ so that $a^{3 k-1} = x^k  , a^{3 k - 2} = y^k, a^{3 k - 3} = z^k$. In fact, due to the construction described in Section 4.2 of \cite{borodin2015random}, which comes from \cite{BorFerr2008DF} and originally stems from \cite{DiaconisFill1990}, we know that along ``down-right" space-time paths $(n_i, T_i)_{i=0}^l$ where $n_i - n_{i-1} \leq 0$ and $T_i-T_{i-1} \geq 0$ the joint distribution of the interlacing arrays $a^{n_i}(T_i)$ is still a Schur process. It follows that the determinantal structure is preserved along these paths as well, and the correlation Kernel is 
 \begin{align}
\label{eqn:kast_time}
K((X_1, U, T_1), (X_2, V, T_2)) = \frac{1}{(2 \pi i)^2} \int_{\Gamma_{\pm}} \int_{\Gamma_{\mp}} \frac{\Phi( 3 k_1 - m_1, T_1; z)}{\Phi( 3 k_2 - m_2, T_2; w)}  \frac{w}{z-w} \; \; \frac{dw dz}{z^{U+1} w^{-V+1}} \;\;.
\end{align}

\end{remark}

\subsection{Inverse Kasteleyn from Correlation Kernel}

We now derive an exact formula for the inverse of the Kasteleyn matrix $\mathbf{K}$ on the tower graph of size $N$. Since the correlation kernel $K$ defined on $\mathfrak{T}_N \times \mathfrak{T}_N$ encodes the same probabilistic data as the inverse Kasteleyn, it is reasonable to believe we can derive a formula for the latter from the former. In this section we prove that this is indeed the case. Using the integral formula \eqref{eqn:kast} for $K$, we can extend its domain definition to arbitrary $((X', U'), (X, U)) \in \mathbb{Z}^2$. We call this the \emph{extended kernel}.

Define the operator $\tilde{K} : V_B^N \rightarrow V_W^N$ as $\tilde{K}(\mathbf{w}(X', U'), \mathbf{b}(X, U)) = K((X', U'),(X, U))$. We claim that $\tilde{K} = \mathbf{K}^{-1}$. To show this we have two lemmas.

\begin{lemma}
If $(X_2, V)$ indexes a black vertex, the extended kernel $K$ satisfies
$$
\begin{cases}
K((X_1+1, U), (X_2, V)) - K((X_1-1, U), (X_2, V)) &  \\ 
\;\;\;+ \beta K((X_1+1, U-1), (X_2, V)) + K((X_1-1, U+1), (X_2, V)) = \delta_{(X_1, U), (X_2, V)}&  X_1 = 3 k - 1\;  \\ 
K((X_1, U), (X_2, V)) - K((X_1-1, U), (X_2, V)) + \alpha K((X_1-1, U+1), (X_2, V)) = \delta_{(X_1, U), (X_2, V)}&  X_1 = 3 k - 2\; \;\;.  \\ 
\end{cases}
$$
\end{lemma}

\begin{proof}
First suppose $X_1 = 3 k - 1$. Then $-K((X_1, U), (X_2, V))$ can be written as 
\begin{align*}
K((X_1, U), (X_2, V)) &=\frac{1}{(2 \pi i)^2} \int_{\Gamma_{\pm}} \int_{\Gamma_{\mp}} (1 + \beta z) \frac{(1 + \beta z)^{N- k} (1-z^{-1})^{k } (1-\alpha z^{-1})^{k}}{\Phi(X_2, N; w)} \; \frac{w dw dz}{(z-w)z^{U+1} w^{-V + 1}} \\
&= K((X_1+1, U), (X_2, V)) + \beta K((X_1+1, U-1), (X_2, V)) \;\;.
\end{align*}
But we also have
\begin{align*}
K((X_1, U), (X_2, V)) &= \frac{1}{(2 \pi i)^2} \int_{\Gamma_{\pm}} \int_{\Gamma_{\mp}} (1-z^{-1}) \frac{(1 + \beta z)^{N- k+1} (1-z^{-1})^{k-1 } (1-\alpha z^{-1})^{k}}{\Phi(X_2, N; w)} \; \frac{w dw dz}{(z-w) z^{U+1} w^{-V + 1}} \\
&=K((X_1-1, U), (X_2, V)) - K((X_1-1, U+1), (X_2, V))  \\
& - \frac{1}{(2 \pi i)^2}\mathbf{1}_{X_1 = X_2}  (\int_{\Gamma_{-}} \int_{\Gamma_{+}} - \int_{\Gamma_{+}} \int_{\Gamma_{-}} )  \frac{\Phi(X_1, N; z)}{\Phi(X_2, N; w)} \; \frac{ w dw dz}{(z-w) z^{U+1} w^{-V + 1}}\;\;.
\end{align*}
Subtracting the second formula from the first, we get that the left hand side of the first case equals
\begin{align*}
&K((X_1+1, U), (X_2, V)) - K((X_1-1, U), (X_2, V))   \\ 
&\;\;\;+ \beta K((X_1+1, U-1), (X_2, V)) + K((X_1-1, U+1), (X_2, V))  \\
&= - \frac{1}{(2 \pi i)^2} \mathbf{1}_{X_1 = X_2}  (\int_{\Gamma_{-}} \int_{\Gamma_{+}} - \int_{\Gamma_{+}} \int_{\Gamma_{-}} )   \; \frac{w dw dz}{(z-w) z^{U+1} w^{-V + 1}} \\
&=   \delta_{(X_1, U), (X_2, V)} \;\;.
\end{align*}

Now for the second case, $X_1 = 3 k - 2$, we have
\begin{align*}
K((X_1, U), (X_2, V)) &=\frac{1}{(2 \pi i)^2} \int_{\Gamma_{\pm}} \int_{\Gamma_{\mp}} (1 - \alpha z^{-1}) \frac{(1 + \beta z)^{N- k+1} (1-z^{-1})^{k -1} (1-\alpha z^{-1})^{k-1}}{\Phi(X_2, N; w)} \; \frac{w dw dz}{(z-w)z^{U+1} w^{-V + 1}}  \\
&= K((X_1-1, U), (X_2, V)) - \alpha K((X_1-1, U+1), (X_2, V))  \\
& -\frac{1}{(2 \pi i)^2}\mathbf{1}_{X_1 = X_2}  (\int_{\Gamma_{-}} \int_{\Gamma_{+}} - \int_{\Gamma_{+}} \int_{\Gamma_{-}} )   \; \frac{ w dw dz}{(z-w) z^{U+1} w^{-V + 1}}\;\;.
\end{align*}
and so again we get that the left hand side of the second case equals
$$ \delta_{(X_1, U), (X_2, V)} \;\;.$$

\end{proof}

Next, we state a lemma which roughly says the following: If $\mathbf{w}(X, U)$ is a boundary vertex, meaning it is not in the tower graph but is exactly one edge away from a vertex belonging to the tower of size $N$, then (slightly abusing notation) $\tilde{K}(\mathbf{w}(X, U), \mathbf{b} (X_2, V)) = 0$. 

\begin{lemma}[Boundary conditions]
Suppose $0 \leq X \leq 3 N$ and $X = 3 k - a$, so that $\mathbf{w}(X, U)$ is a white vertex of the square hexagon lattice. In each of the following three cases
\begin{enumerate}
\item $U = N - k + 1$ 
\item $a = 2$ and $U = - 2 k$
\item $a = 3$ and $U = - 2 k +1 $
\end{enumerate}
we have
$$K((X, U), (X', U'))  = 0 \;\;.$$
\end{lemma}

\begin{proof}
The proof proceeds by plugging in these values for $U$ case by case and analyzing residues to argue that the double integral is $0$. 
\begin{enumerate}
\item Then, WLOG assuming for this case that $X = 3 k -1$,
\begin{align*}
K((X, U), (X_2, V)) &= \frac{1}{(2 \pi i)^2} \int_{\Gamma_{\pm}} \int_{\Gamma_{\mp}} \frac{(1 + \beta z)^{N- k+1} (1-z^{-1})^{k } (1-\alpha z^{-1})^{k}}{\Phi(X_2, N; w)} \; \frac{w dw dz}{(z-w)z^{N-k+2} w^{-V + 1}}\;\;.
\end{align*}
First suppose $X \geq X_2$. Then we can expand the $z$ contour $\Gamma_+$ out to infinity, and since the integrand decays as $|\frac{(1 + \beta z)^{N- k+1}}{(z-w)z^{N-k+2}}| = O(\frac{1}{|z|^2})$, the residue at $z = \infty$ is $0$, so the value of the $z$ integral is $0$ (for each $w$ on the contour). If $X < X_2$, then we pick up the $z = w$ residue in the process of dragging the $z$ contour out to $\infty$. So we must argue that 
\begin{align*}
 \int_{\Gamma_{+}} \frac{(1 + \beta w)^{N- k+1} (1-w^{-1})^{k } (1-\alpha w^{-1})^{k}}{(1+\beta w)^{N-k_2+1} (1-w^{-1})^{k_2 -1 +\mathbf{1}_{m_2 =1}} (1-\alpha w^{-1})^{k_2 -1+ \mathbf{1}_{m_2 \leq 2}} } \; \frac{dw}{w^{N-k+2} w^{-V}} = 0 \;\;.
\end{align*}
Now, since $V \leq N-k_2$, we can see that the integrand behaves as $O(\frac{1}{|w|^2})$ for $|w|$ large. Also, there is no residue at $-\frac{1}{\beta}$ because $X < X_2$ and so $k \leq k_2$. Thus, the integral again is $0$ because the residue at $\infty$ is $0$.

\item Suppose $a = 2$ and $U = - 2 k$. Then 
\begin{align*}
K((X, U), (X_2, V)) &= \frac{1}{(2 \pi i)^2} \int_{\Gamma_{\pm}} \int_{\Gamma_{\mp}} \frac{(1 + \beta z)^{N- k+1} (1-z^{-1})^{k -1 } (1-\alpha z^{-1})^{k}}{\Phi(X_2, N; w)} \; \frac{w dw dz}{(z-w)z^{-2 k + 1} w^{-V + 1}}\;\;.
\end{align*}
If $X < X_2$, we can shrink the $z$ contour to $0$ without crossing any residues, and since $(1-z^{-1})^{k-1} (1-\alpha z^{-1})^{k}z^{2 k -1} = O(1)$ as $z \rightarrow 0$, there is no residue at $0$. Thus the integral is $0$. If $X \geq X_2$, then we again would like to shrink the $z$ contour to $0$, but we must also consider the residue at $z = w$. The result from this residue, modulo prefactors, is 
\begin{align*}
 \int_{\Gamma_{-}} \frac{(1 + \beta w)^{N- k+1} (1-w^{-1})^{k -1} (1-\alpha w^{-1})^{k}}{(1+\beta w)^{N-k_2+1} (1-w^{-1})^{k_2 -1 +\mathbf{1}_{m_2 =1}} (1-\alpha w^{-1})^{k_2 -1+ \mathbf{1}_{m_2 \leq 2}} } \; \frac{dw}{w^{-2k+1} w^{-V}} \;\;.
\end{align*}
Now since $X \geq X_2$, there is no residue at $1$ or $\alpha$. Also, since $V + (k_2 -1+ \mathbf{1}_{m_2 \leq 2}) + (k_2 -1 +\mathbf{1}_{m_2 =1}) = V + 2 k_2 -1 + \mathbf{1}_{m_2 =1} \geq 0$, by the same argument as above the integrand is $O(1)$ as $w \rightarrow 0$, so the residue at $w= 0$ is $0$.
\item This case is very similar to case (2), so we omit details.

\end{enumerate}

\end{proof}

\begin{proposition}\label{prop:invK}
 $$\tilde{K} = \mathbf{K}^{-1}$$ is the inverse Kasteleyn. 
\end{proposition}
\begin{proof}
The proposition follows immediately from the two lemmas. 
\end{proof}

\section{Local Statistics}
\label{sec:loc}

\subsection{Translation Invariant Gibbs Measures and Complex Coordinate}
\label{subsec:gibbs_meas}

 Now we review the description of the translation invariant Gibbs measures of the dimer model on a $\mathbb{Z}^2$ periodic lattice given in \cite{KOS2006}, and describe them in our setting. In the following general discussion, we consider a graph $G$ which is invariant under a $\mathbb{Z}^2$ action.

\begin{figure}
\includegraphics[scale=0.8]{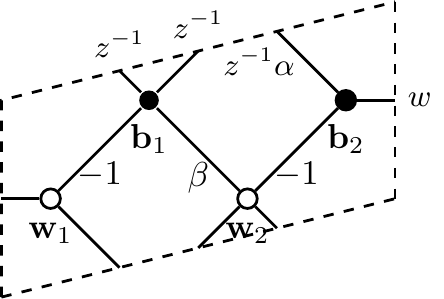}
\caption{Above is the fundamental domain we choose for the square hexagon graph $G$. We also view this as an embedding of the graph $G_1 = G / \mathbb{Z}^2$ on a torus. The weights shown are used to construct $\mathbf{K}_1(z, w)$.}
\label{fig:mag}
\end{figure}

 \begin{figure}
\includegraphics[scale=0.8]{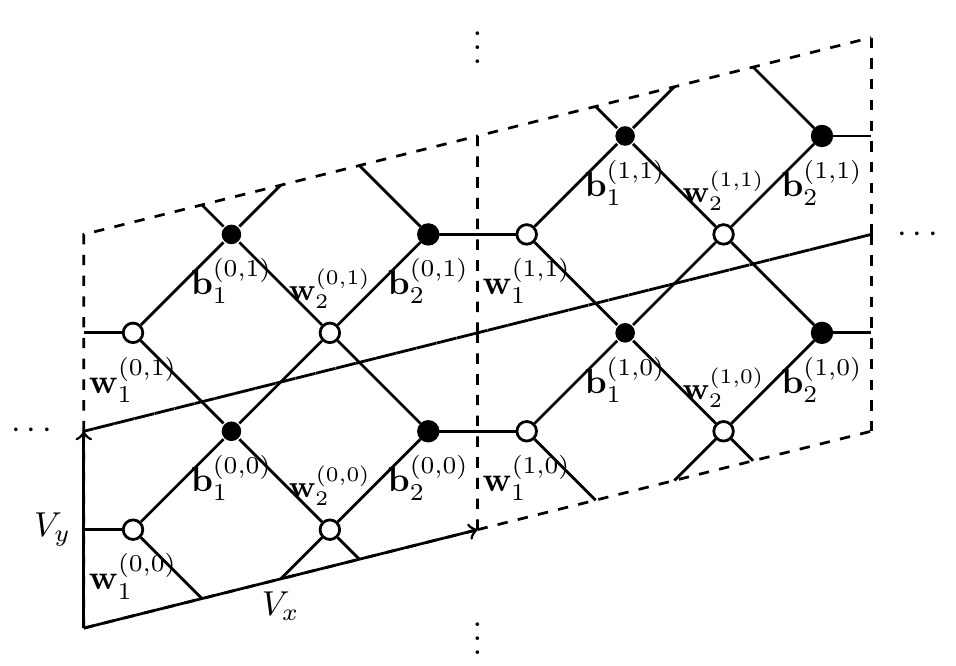}
\caption{Above is the infinite square hexagon graph $G$, with coordinates shown for some vertices. The $\mathbb{Z}^2$ action is generated by the translations by $V_x$ and $V_y$. If $F$ denotes the position of a face, and $H$ is the height function of a dimer configuration randomly sampled from a Gibbs measure (expectation with respect to which we denote as $E$), we have $s = E[H(F + V_x) - H(F)], t = E[H(F + V_y) - H(F)]$.}
\label{fig:z2lattice}
\end{figure}

An important object in the study of dimers on $\mathbb{Z}^2$ invariant lattices is the \emph{characteristic polynomial} of the graph, which is a Laurent polynomial determined by the edge-weighted graph $G$, and is denoted by $P(z, w)$. As stated in Theorem 2.1 of \cite{KOS2006}, there is a two parameter family of ergodic, translation invariant Gibbs measures on dimer covers of $G$ corresponding to pairs of allowed average height function slopes $(s, t)$. Here ergodicity and translation invariance are both with respect to the $\mathbb{Z}^2$ action, and ``allowed slopes" means that $(s, t)$ is in the Newton polygon of $P$. These Gibbs measures can be in three phases, the \emph{liquid phase}, the \emph{frozen phase}, and the \emph{gas phase}. The set of Gibbs measures can also be parameterized by \emph{magnetic fields} $(B_x, B_y)$, a pair of real numbers in the \emph{amoeba} of the algebraic curve $\{P(z, w) = 0\}$. The amoeba is defined by the image of the curve $\{P(z, w) = 0\}$ under the map $(z, w) \mapsto (\log |z|, \log |w|)$. For $(B_x, B_y)$ in the interior of the amoeba, this correspondence is one-to-one, and yields the set of liquid phase Gibbs measures.

 The \emph{surface tension} is the quantity 
$$\sigma(s, t) \defeq \lim_{N \rightarrow \infty} \frac{1}{N^2} \log Z_{N, s, t} $$
where $ Z_{N, s, t}$ is the partition function of the dimer model on the $N \times N$ torus graph $G / (N \mathbb{Z})^2$ with the state space restricted so that allowed configurations must have average height changes of approximately $\floor{N s}$ around a loop in the $x$ direction and $\floor{N t}$ around a loop in the $y$ direction. The quantity $\sigma(s, t)$ is known to exist and be a strictly convex function of $(s, t)$. It is well known that the magnetic fields correspond to average slopes $(s, t)$ of the height function via 
$$B_x = \frac{\partial \sigma}{\partial s}, B_y = \frac{\partial \sigma}{ \partial t } \;\;.$$ 
For each pair of magnetic fields $(B_x, B_y)$, there is an inverse of the Kasteleyn matrix on $G$ from which one can explicitly construct an ergodic Gibbs measure of the corresponding slopes. We will illustrate this procedure concretely in our setup.

 In our case, similarly to the well studied cases of lozenge and domino tilings with uniform weights, one can also parameterize Gibbs measures by a \emph{complex coordinate} $z$ in the the upper half plane $\mathbb{H}$. More precisely, given $z \in \mathbb{H}$, there is a unique $w$ such that $P(z,w) = 0$, and then the Gibbs measure corresponding to $z$ is the liquid phase Gibbs measure corresponding to magnetic fields $B_x = \log |z|, B_y = \log |w|$. When restricted to $z \in \text{int}(\mathbb{H})$, this gives a one-to-one correspondence to the liquid phase Gibbs measures.

Now we explicitly define the Gibbs measure corresponding to a complex coordinate $z_0$ in our situation. To do this, we first construct the inverse Kasteleyn of the Gibbs measure corresponding to magnetic fields $B_x = \log r_1, B_y = \log r_2$, and then we simplify it and see that we get a Gibbs measure for each $z_0 \in \mathbb{H}$. From here onwards, $G$ (which we called $\mathcal{L}$ in the introduction) refers to the infinite square hexagon graph, shown in Figure \ref{fig:z2lattice}.

We define the Kasteleyn matrix $\mathbf{K}$ of the infinite graph $G$ by 
\begin{align*}
\mathbf{K}(\mathbf{b}, \mathbf{w}) &= 
\begin{cases}
\text{sign}(\mathbf{b}, \mathbf{w}) \cdot \text{wt}(\mathbf{b}, \mathbf{w}) & \text{if $(\mathbf{b}, \mathbf{w})$ is an edge } \\
0 & \text{otherwise}
\end{cases}
\end{align*}
i.e. as the product of the Kasteleyn sign with the edge weight on the infinite graph $G$, shown in Figure \ref{fig:SHL}.

As $\mathbf{K}$ is infinite dimensional, its inverse is not unique, and we will present the construction of a two parameter family of inverses corresponding to the ergodic Gibbs measures in our case (see \cite{KOS2006}). The \emph{magnetically altered} Kasteleyn matrix for $G_1 = G/\mathbb{Z}^2$ in our case is
\begin{align}
\mathbf{K}_1(z,w) = \begin{blockarray}{ccc}
\mathbf{w}_1 & \mathbf{w}_2 &  \\
\begin{block}{(cc)c}
  -1+  z^{-1} &z^{-1}+ \beta &  \mathbf{b}_1 \\
  w & -1+  \alpha z^{-1} & \mathbf{b}_2 \\
\end{block}
\end{blockarray} \;\;.
\label{eqn:tkast}
\end{align}
It is obtained by putting appropriate complex weights $z, z^{-1}$ and $w, w^{-1}$ on edges crossing fundamental domains in the vertical and horizontal directions, respectively. See Figure \ref{fig:mag}.

The \emph{characteristic polynomial} of our graph is 
\begin{equation}
\label{eqn:charpoly} 
P(z,w) \defeq \det \mathbf{K}_1(z,w) = (1- \alpha z^{-1})(1- z^{-1}) - w (z^{-1}+\beta) \;\;.
\end{equation}
We have $\mathbf{K}_1(z,w)^{-1}_{\mathbf{w}_i, \mathbf{b}_j} = \frac{Q(z,w)_{\mathbf{w}_i, \mathbf{b}_j}}{P(z,w)}$ where 
\begin{equation}
Q(z,w) \defeq 
\begin{blockarray}{ccc}
\mathbf{b}_1 & \mathbf{b}_2 &  \\
\begin{block}{(cc)c}
  -1+ \alpha z^{-1} & -z^{-1}- \beta &  \mathbf{w}_1 \\
  -w & -1+ z^{-1} &  \mathbf{w}_2 \\
\end{block}
\end{blockarray} \;\;.
\end{equation}

Let $\mathbf{w}_a^{(T,Y)}$ and $\mathbf{b}_a^{(T,Y)}$, $(T, Y) \in \mathbb{Z}^2, a \in \{0,1\}$, denote the black and white vertices of the lattice, where $(T, Y)$ index the translates of the fundamental domain, and $a$ is used to index the vertex inside each fundamental domain. One may check directly that for any $(r_1, r_2)$, an inverse of $\mathbf{K}$ can be written as 
\begin{align}
\label{eqn:GibbsK}
K^{r_1,r_2}[\mathbf{w}_i^{(T',Y')}, \mathbf{b}_j^{(T,Y)}] &= 
\frac{1}{(2 \pi i)^2} \int_{|w| = r_2} \int_{|z| = r_1} \frac{Q(z,w)_{\mathbf{w}_i, \mathbf{b}_j}}{P(z,w)} w^{T'-T} z^{Y-Y'} \; \frac{d z}{z} \frac{d w}{w}  \;\;.
\end{align}
Furthermore, it is well known that this inverse Kasteleyn defines an ergodic Gibbs measure determined by the determinantal formula: For $M$ sampled from the measure, and edges $(w_1,b_1),\dots, (w_k,b_k) $, 

\begin{equation}
\text{Prob}((w_1,b_1),\dots, (w_k,b_k) \in M) =
\left( \prod_{j=1}^k \mathbf{K}[b_j, w_j] \right) \det(K^{r_1,r_2}[w_i, b_j])_{i,j=1}^k \;\;.
\label{eqn:detcor}
\end{equation}
Now we will massage the expression \eqref{eqn:GibbsK}. This will make clear the correspondence between Gibbs measures and the complex coordinate.

Define
$$w(z) \defeq \frac{(1-z^{-1})(1-\alpha z^{-1}) }{z^{-1} + \beta} $$
and define $\gamma = \gamma(i, j), \theta = \theta(i, j), \delta = \delta(i, j)$, for $i, j = 1,2$ by 
\begin{align*}
((\gamma, \theta, \delta))_{i,j = 1,2} = 
\begin{pmatrix}
(0, -1, 0) & (-1, -1, -1) \\
(0, 0 , 1) &  (-1, 0, 0)
\end{pmatrix} \;\;.
\end{align*}

By doing the $w$ integration first and taking the residues at $w = w(z)$ and $w = 0$, we write the kernel above as follows:

\begin{align*}
&K^{r_1,r_2}[\mathbf{w}_i^{(T',Y')}, \mathbf{b}_j^{(T,Y)}] \\
&= 
-\frac{1}{2 \pi i}\int_{\overline{z_0}}^{z_0} \frac{Q(z, w(z) )_{\mathbf{w}_i, \mathbf{b}_j}}{ (1-z^{-1})(1-\alpha z^{-1}) } w(z) ^{T'-T} z^{Y-Y'} \; \frac{d z}{z}  + \int_{|z| = r_1} \text{Res($w = 0$)} \; dz \;\;.
\end{align*}
Here $z_0 \in \mathbb{H}$ is the unique point in $\mathbb{H} \cap \{ |z| = r_1 \}$ such that $P(z_0, w_0) = 0$ for some $|w_0| = r_2$, and $\int_{\overline{z_0}}^{z_0}$ denotes integration along the part of the contour $|z| = r_1$ with $|w(z)| \leq r_2$. We can check that this part of the contour is counterclockwise and crosses the real axis to the right of the origin. The residue from $0$ only contributes when $X' < X$, where $X, X', U, U' $ are defined by $\mathbf{w}_i^{(T',Y')} = \mathbf{w}(X', U')$, and $\mathbf{b}_j^{(T,Y)} = \mathbf{b}(X, U)$. Furthermore, in this case the resulting integrand is exactly equal to that of the first term, so if the residue contributes there is a cancellation and we instead can integrate along a clockwise path from $\overline{z_0}$ to $z_0$ which crosses $\mathbb{R}$ to the left of the origin. Thus, simplifying, the integral can be rewritten as
\begin{align*}
\frac{1}{2 \pi i  } \int_{\overline{z_0}}^{z_0} \frac{ (1- z^{-1})^{\gamma} (1- \alpha z^{-1})^{\theta} }{(1+ \beta z)^{\delta}} \left(\frac{(1-  z^{-1}) (1-\alpha z^{-1})}{1 + \beta z  }\right)^{T'-T} z^{\delta}  z^{(Y - T)-(Y'  - T')-1} d z 
\end{align*}
where contour is chosen counterclockwise along the part of $\{|z| = r_1\}$ to the right of the origin if $X' \geq X$ (when the residue does not contribute), and clockwise to the left of the origin otherwise. 

We rewrite this expression in terms of $(X, U)$, $(X', U')$, where $\mathbf{w}_i^{(T',Y')} = \mathbf{w}(X', U')$ and  $\mathbf{b}_j^{(T,Y)} = \mathbf{b}(X, U)$. If $X  = 3 t - m, X' = 3 t' - m'$, then we can check case by case that $T'-T +\theta = t' - t -  \1 \{ m' = 3\}$, $T'-T +\gamma = t' - t  - \1 \{m = 1\}$, and $T' - T + \delta = t' - t$. So this expression for $K^{r_1,r_2}$ can be re-written as 
\begin{align}
\begin{split}
 \frac{1}{2 \pi i  } \int_{\overline{z_0}}^{z_0} (1+ \beta z)^{t-t'} (1-  z^{-1})^{t'-t + \mathbf{1}\{ m' = 1 \} - \mathbf{1}\{ m = 1 \}}   
 (1-\alpha z^{-1})^{t'-t -\mathbf{1}\{m' = 3\} + \mathbf{1}\{m = 3\} } z^{U-U'-1} d z \;\;.
 \end{split}
 \label{eqn:kz0}
\end{align}
This form makes explicit the correspondence between a complex coordinate $z_0 \in \mathbb{H}$ and a Gibbs measure of our model, and also provides a form of the kernel which we will be able to easily identify with our bulk limit. For arbitrary $z_0 \in \mathbb{H}$, we will also denote the same kernel by 
$$K^{z_0}[\mathbf{w}(X', U'), \mathbf{b}(X, U)] \;\;.$$

If $z_0 \in \partial \mathbb{H} = \mathbb{R}$, the kernel $K^{z_0}[\mathbf{w}(X', U'), \mathbf{b}(X, U)]$ is still well defined and corresponds to a frozen Gibbs measure. However, the correspondence between $z_0 \in \mathbb{R}$ and Gibbs measures is no longer one to one. The regions $(-\infty, -\frac{1}{\beta}),  (-\frac{1}{\beta}, 0), (0, \alpha), (\alpha, 1), (1, \infty)$ on the real axis correspond to the $5$ different non-random frozen Gibbs measures, where all correlation functions are identically equal to $0$ or $1$. These are exactly the Gibbs measures whose slopes $(s, t)$ are integer pairs on the boundary of the Newton polygon.

We summarize the above discussion with the following lemma:
\begin{lemma}
\label{lem:gibbsmeas}
For any $z_0 \in \mathbb{H}$, the formula 
\begin{align*}
&K^{z_0}[\mathbf{w}(X', U'), \mathbf{b}(X, U)]\\
&=  \frac{1}{2 \pi i  } \int_{\overline{z_0}}^{z_0} (1+ \beta z)^{t-t'} (1-  z^{-1})^{t'-t + \mathbf{1}\{ m' = 1 \} - \mathbf{1}\{ m = 1 \}}   
 (1-\alpha z^{-1})^{t'-t -\mathbf{1}\{m' = 3\} + \mathbf{1}\{m = 3\} } z^{U-U'-1} d z 
\end{align*}
where the integration contour crosses the real axis in $(0, \infty) \setminus \{1, \alpha\}$ for $X' \geq X$ and in $(-\infty, 0) \setminus \{-\frac{1}{\beta}\}$ for $X' < X$, defines a correlation kernel for an ergodic translation invariant Gibbs measure via the formula \eqref{eqn:detcor}. This Gibbs measure corresponds to magnetic fields $\log |z_0|, \log |w(z_0)|$. The family of kernels $K^{z_0}$ give a parameterization of liquid phase Gibbs measures as $z_0$ ranges over $\text{int}(\mathbb{H})$. For $z_0 \in \mathbb{R} \setminus \{-\frac{1}{\beta}, 0, 1, \alpha\}$, this gives the set of nonrandom frozen Gibbs measures corresponding to integer slopes $(s, t)$ on the boundary of the Newton polygon.
\end{lemma}

 \subsection{Asymptotics: Bulk Limit} 
 \label{subsec:BL}
 
 Now we take asymptotics of the kernel \ref{eqn:kast} in order to show that the local statistics of random tilings of the tower of size $N$ converge to the translation invariant ergodic Gibbs measures, as conjectured in \cite{KOS2006} to be true for general boundary conditions. Recall the rescaled domain of the tower $\mathfrak{T} \coloneqq \{(x, u) \in \mathbb{R}^2 : 0 < x < 1 , - 2 x < u < 1- x \}$. Let $(x, u) \in \mathfrak{T}$ denote fixed macroscopic coordinates. We show that the inverse Kasteleyn evaluated at a pair of vertices which are a finite distance from the vertex $(3 \floor{N x}, \floor{N u})$ at the lattice scale converges to a corresponding entry of the inverse Kasteleyn of a translation invariant ergodic Gibbs measure. An important function in our analysis is the \emph{action function} 
 \begin{equation}
S(z ; x, u) \defeq x (\log(1 - \alpha_1 1/z) + \log(1 - \alpha_2 1/z)) - u \log(z) + (1 - x) \log(1 + \beta z).
\label{eqn:S}
\end{equation}
We have stated Theorem \ref{thm:mainthm} with $\alpha_1 = 1, \alpha_2 = \alpha$, and ultimately we will assume this is the case, but for the analysis of the action function in the proof we will keep both parameters and assume that $\alpha_1 < \alpha_2$ for convenience.

Define the \emph{liquid region} $\mathfrak{L} \subset \mathfrak{T}$ as the subset of $(x, u)$ in the rescaled tower domain such that $S(z ; x, u)$ has a pair of non-real, complex conjugate critical points. The complement of $\overline{\mathfrak{L}}$ in $\mathfrak{T}$, where $S(z ; x, u)$ has distinct real critical points, will be called the \emph{frozen region}. In particular, the \emph{arctic curve} $C = \partial \mathfrak{L}$ separating the two regions is an algebraic curve. 

 \begin{theorem}
 Let $(x, u) \in \mathfrak{L}$. Then
 \begin{equation}\label{eqn:localstat}
 \lim_{N \rightarrow \infty} K((3 \floor{N x} + X', \floor{N u} + U'), (3 \floor{N x} + X, \floor{N u} + U))
  = K^{z_0}[\mathbf{w}(X', U'), \mathbf{b}(X, U)] 
 \end{equation}
where $z_0 = z_0(x, u)$ is the unique critical point in the upper half plane of the function 
$$S(z ; x, u) \defeq x (\log(1 - \alpha_1 1/z) + \log(1 - \alpha_2 1/z)) - u \log(z) + (1 - x) \log(1 + \beta z) \;\;.$$
In particular, the local limit is the infinite volume Gibbs measure with complex coordinate $z_0$. In the region where the critical points of $S$ are distinct and real, there is an extension of the function $z_0(x, u) \in \mathbb{R}$ such that \eqref{eqn:localstat} still holds. In particular, in this region the local limit is one of the $5$ nonrandom frozen Gibbs measures.
\label{thm:loc}
 \end{theorem}

\begin{figure}
\includegraphics[scale=0.6]{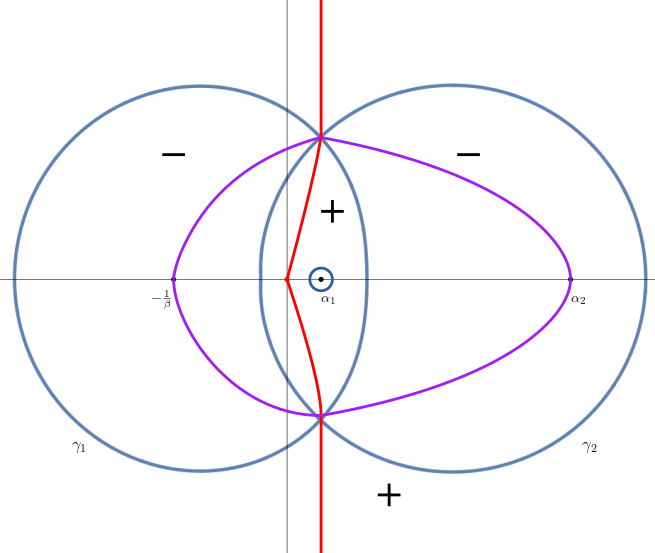}\\
\caption{Shown is a schematic depiction of the level lines of $\text{Re}(S(z)) = \text{Re}(S(z_c)) $ (blue), the level lines $\text{Im}(S(z)) = \text{Im}(S(z_c))$ and their reflections about the real axis (red and purple). The plus and minus symbols indicate regions where $\text{Re}(S(z)) - \text{Re}(S(z_c))$ is positive and negative, respectively.}
\label{fig:lcs_main}
\end{figure}

\begin{figure}
\begin{subfigure}[b]{0.3\textwidth}
         \centering
         \includegraphics[scale=0.25]{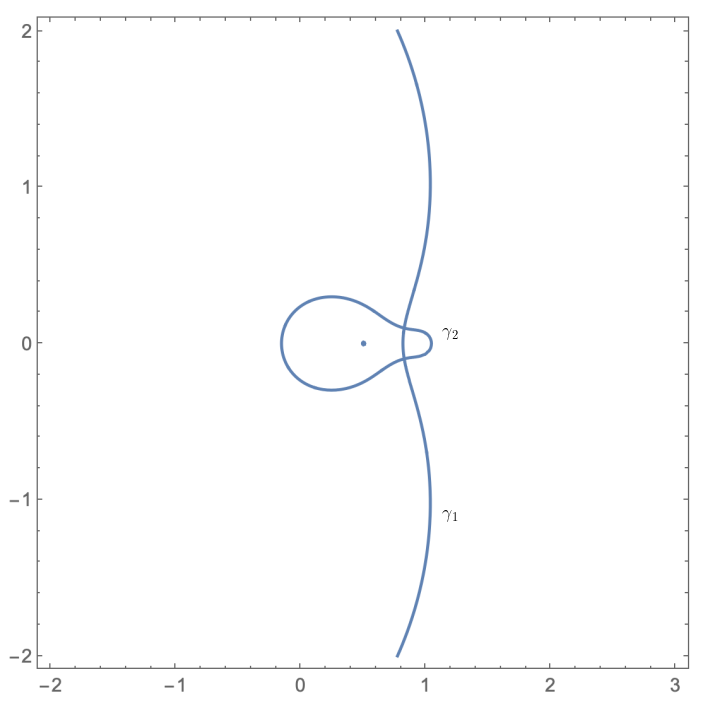}
         \caption{}
         \label{fig:levseta}
     \end{subfigure}
     \hspace{40 pt}
     \begin{subfigure}[b]{0.3\textwidth}
         \centering
         \includegraphics[scale=0.2]{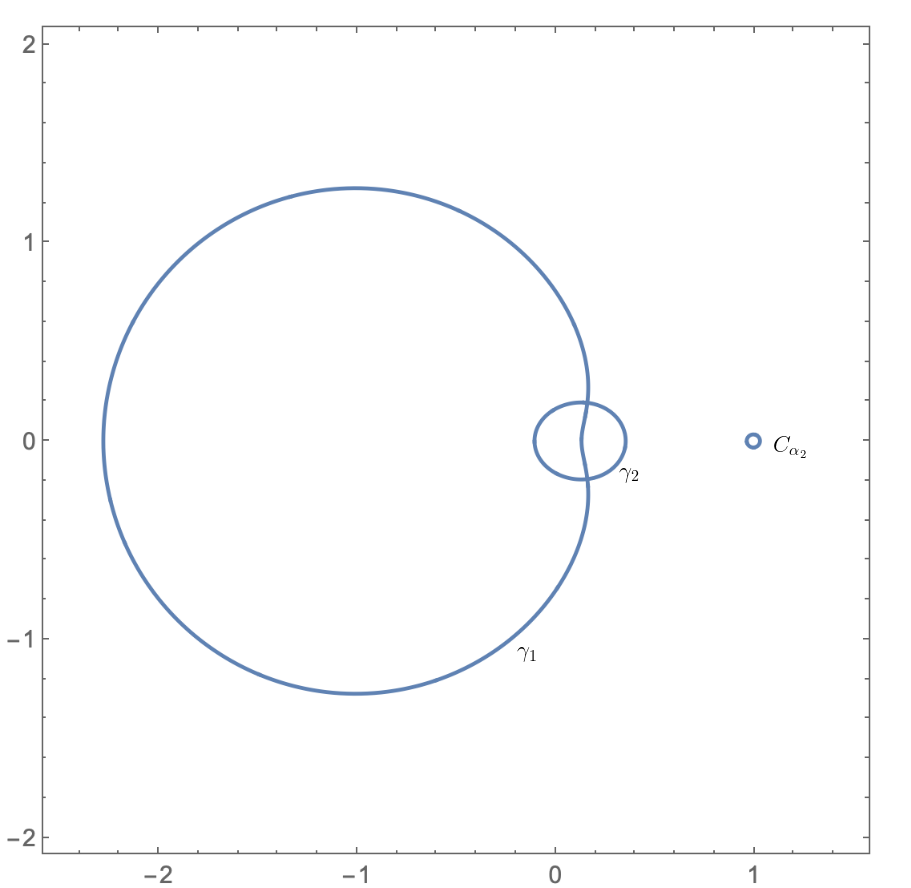}
         \caption{}
         \label{fig:levsetb}
     \end{subfigure}
\caption{In (A) we show an example with $\alpha_1 = 0.5$, $\alpha_2 = 1$, $\beta = 1$, to illustrate a generic example. Here $\gamma_1$ is indeed a closed loop, and there is another small loop around $\alpha_1 = 0.5$. In this case we would have $\gamma_1 = \gamma_z, \gamma_2 = \gamma_w$.
In (B) $\alpha_1 = 0.2$, $\alpha_2 = 1$, $\beta = 1$. This is a case where there is an extra disjoint loop $C_{\alpha_2}$ around $\alpha_2$, and in this case we set $\gamma_w = \gamma_2 \bigcup C_{\alpha_2}$. In both cases $(x, u)$ is in the liquid region.}
\end{figure}

\begin{figure}
\centering
\includegraphics[scale=0.4]{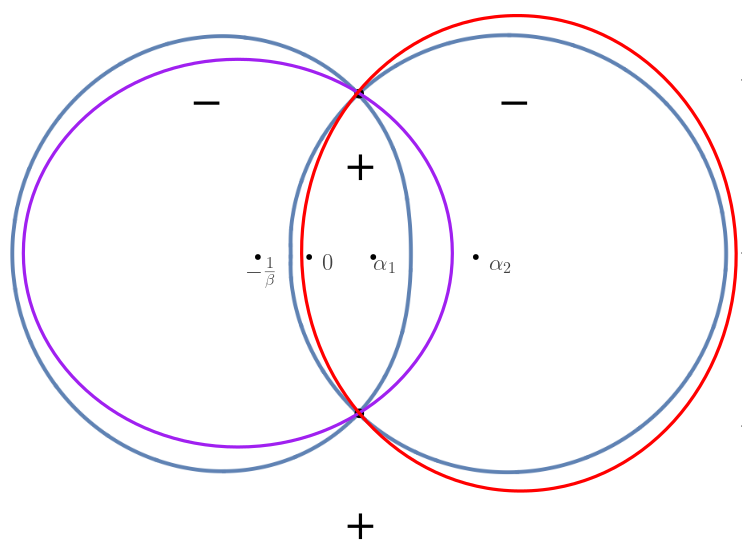}
\caption{Above we illustrate the deformed contours $\tilde{\gamma}_z$ (purple) and $\tilde{\gamma}_w$ (red). Level lines are in blue. Not all level lines are shown.}
\label{fig:def2}
\end{figure}

\begin{figure}
      \begin{subfigure}[b]{0.6\linewidth}
         \centering
         \includegraphics[scale=0.4]{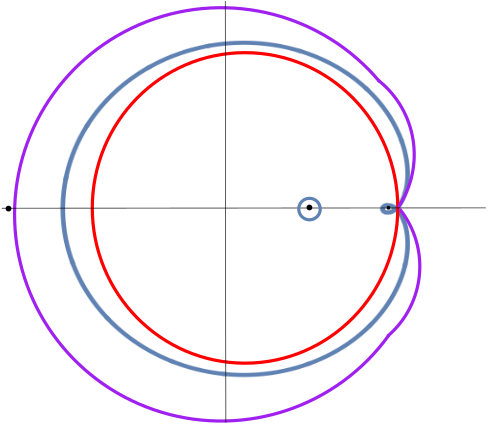}
         \caption{$z_0 \in (\alpha_2, \infty)$}
         \label{fig:F3}
     \end{subfigure}
\caption{The level curves are in blue, $\gamma_z$ is in purple, and $\gamma_w$ is in red, and the poles $-\frac{1}{\beta}, \alpha_1, \alpha_2$ are shown on the real axis. Note that in each case here $\gamma_z \subset \text{exterior}(\gamma_w)$.}
\label{fig:frozen_regions_gz}
\end{figure}

\begin{proof}

We perform the proof in steps. We use a variant of the method of steepest descent, which is fairly standard. First, we analyze the liquid region, which is the set of points $(x, u)$ for which not all critical points of $S$ are real.

\textbf{Step 1:} First, write the integrand in \ref{eqn:kast} as $\frac{1}{z-w}\exp(N (S(z) - S(w)) + O(1))$ with $S(z) = S(z ; x, u) \defeq x (\log(1 - \alpha_1/z) + \log(1 - \alpha_2 1/z)) - u \log(z) + (1 - x) \log(1 + \beta z)$. Assume WLOG that $\alpha_1 < \alpha_2$. Here $O(1)$ stands for a function of $(z, w)$ which, uniformly for all $(z, w)$ in a suitably chosen compact subset which excludes a small neighborhood of the poles, has absolute value upper bounded by a constant $C$ for all $N$ large enough.

\textbf{Step 2:} We will now describe a deformation of the integration contours, such that the integration along the new contours will vanish as $N \rightarrow \infty$. Changing the contours will require us to add a correction term due to a new residue at $z = w$ picked up from the deformation, and this correction will be exactly the term we want (see Step 3 below).

First, by our assumption on $(x, u)$, and because $\frac{\partial}{\partial z} S(z) = 0$ is a cubic equation, there is a unique complex conjugate pair of critical points. We pick the one which occurs in the upper half plane and call it $z_c = z_c(x, u)$. Near $z_c$, there are two branches of the level set $\text{Re}(S(z)) = \text{Re}(S(z_c))$ intersecting orthogonally. Since $S(z)$ is real analytic and $|\text{Re}(S(z))| \rightarrow \infty$ as $|z| \rightarrow \infty$, each branch forms a closed loop which is symmetric about the real axis, and these two loops intersect at $z_c$ and $\overline{z_c}$. One of these loops may cross itself if it passes through the third critical point of $S$ on the real axis, which necessarily occurs in the interval $(\alpha_1, \alpha_2)$. However, this does not change the argument, so we first proceed assuming that this does not occur, and later we indicate what happens in this case.

There are four points $x_1 < x_2 < x_3 < x_4$ at which these two level curves intersect the real axis. One of these two level lines, call it $\gamma_1$, intersects the real axis at $x_1$ and $x_3$, and the other, $\gamma_2$, intersects at $x_2$ and $x_4$. We claim that $x_1 < -\frac{1}{\beta} < x_2 < 0 < x_3 < x_4$, and at least one of $\alpha_1, \alpha_2$ belongs to the interval $(x_3, x_4)$. This follows from properties of the function $S(z)$. Indeed, the function $\text{Re}(S(z))$ has the following behavior for $z \in \mathbb{R}$: It decreases from $\infty$ to $-\infty$ on $(-\infty, -\frac{1}{\beta})$ and $(0, \alpha_1)$, and increases from $-\infty$ to $\infty$ on $(-\frac{1}{\beta}, 0)$, $(\alpha_2, \infty)$, and in $(\alpha_1, \alpha_2)$ it starts at $-\infty$ near $\alpha_1$, increases to a local max (this is the third critical point), and then decreases back to $-\infty$ near $\alpha_2$. Thus, there are either $4$ or $6$ zeros of $\text{Re}(S(z)) - \text{Re}(S(z_c))$ on the real axis. If there are $4$ zeros then we are done. If there are $6$, the only cases that we have to rule out involve a situation where the union of the contours $\gamma_1 \cup \gamma_2$ has two zeros in the interval $(\alpha_1,\alpha_2)$. But if this is the case, then there exists a compact set bounded by level lines of $\text{Re}(S(z))$ (enclosing the interval between these two zeros) on which $\text{Re}(S(z))$ is harmonic. But this is impossible, as $\text{Re}(S(z))$ is not constant. Therefore, the properties of $\text{Re}(S(z))$ we stated above imply the claim about the positions of the $x_i$. See Figure \ref{fig:lcs_main} for an illustration. 

Now we define contours $\gamma_z, \gamma_w$. We let $\gamma_z \defeq \gamma_1$. There are cases for $\gamma_w$:
\begin{enumerate}
\item If $\gamma_2$ contains $\alpha_1$ and $\alpha_2$, set $\gamma_w = \gamma_2$. 
\item On the other hand $\gamma_2$ only contains $\alpha_1$, then there is another disconnected component of the level curve encircling $\alpha_2$, and we set $\gamma_w$ to be the union of $\gamma_2$ and this component. See Figures \ref{fig:levseta} and \ref{fig:levsetb}, respectively. The case when $\text{Re}(S(z)) = \text{Re}(S(z_c))$ contains the third critical point can be seen as a limit where $\gamma_2$ and the component encircling $\alpha_2$ meet on the real axis.
\end{enumerate}

  Since $-\frac{1}{\beta} < x_2$, $\gamma_w$ always excludes the residue at $-\frac{1}{\beta}$, but it contains $\alpha_1$ and $\alpha_2$ in its interior by construction. Both $\gamma_z$ and $\gamma_w$ contain $0$ in their interior.

First, we deform the $w$ contour to $\gamma_w$, and the $z$ contour to $\gamma_z$. This deformation of contours does not cross any residues of the integrand of the kernel (see equation \eqref{eqn:kast}), except for the one at $z = w$. We momentarily ignore the contribution of the $z = w$ residue, and further deform the contours, $\gamma_z \rightarrow \tilde{\gamma}_z$, $\gamma_w \rightarrow \tilde{\gamma}_w$ (see Figure \ref{fig:def2}), without crossing any additional residues, so that for each $(z, w) \in \tilde{\gamma}_z \times \tilde{\gamma}_w \setminus \{z_c, \overline{z_c}\}$, $\text{Re}(S(z) - S(w))  < 0$. 

With these contours, the exponential decay of $\exp(N (S(z) - S(w)))$ forces a uniform decay of the integrand on all but a small neighborhood of the critical points, so the integral vanishes. More precisely, we can choose $0<\epsilon< \frac{1}{2}$ and replace the integral over these contours with the same integral restricted to the part of the contours in an $\frac{1}{N^{\frac{1}{2}-\epsilon}}$-neighborhood of $z_c, \overline{z_c}$, at the cost of an error which is $O(N^{-k})$ for any $k$ as $N \rightarrow \infty$. It is then standard to check, by making the integration variable change $\tilde{w} = \sqrt{N}(w - z_c), \tilde{z} = \sqrt{N}(z - z_c)$, that the integral indeed vanishes as $N \rightarrow \infty$. 

\textbf{Step 3:} The remaining contribution in the $N \rightarrow \infty$ limit is the residue we pick up at $z = w$ from changing the contours to $\gamma_z$ and $\gamma_w$. This residue consists of a contour integral along a part of $\gamma_z$ connecting $\overline{z_c}$ to $z_c$. If $X' \geq X$, then we get a residue on the part of $\gamma_z$ crossing the real axis to the right of $0$, and if $X' < X$ it is the part crossing to the left of $0$. Now, going back to the formula \eqref{eqn:kast}, and plugging in $(3 \floor{N x} + X', \floor{N u} + U'), (3 \floor{N x} + X, \floor{N u} + U)$ and setting $z = w$ in the integrand, we see that our result is exactly the formula for the Gibbs measure in lemma \ref{lem:gibbsmeas}, with $z_0 = z_c$.

That concludes the argument for the liquid region. Now we analyze the \textbf{frozen region}, which means that for our choice of $(x, u)$ there are three distinct real critical points of $S$. It is still true that there is always at least one critical point in the interval $(\alpha_1, \alpha_2)$. The other two critical points both lie in one of the $5$ intervals
$(-\infty, -\frac{1}{\beta}), (-\frac{1}{\beta}, 0), (0, \alpha_1), (\alpha_1, \alpha_2), (\alpha_2, \infty)$. In order to match the local limit with the frozen ergodic Gibbs measure $K^{z_0}$ in each of these five cases, it suffices show that there exist contours $\gamma_z, \gamma_w$ along $\text{Re}(S(z)) = \text{Re}(S(z_c))$, such that $\gamma_w$ contains $\alpha_1,\alpha_2$ and excludes $\frac{-1}{\beta}$, and both contours contain $0$, and also with the property that
\begin{itemize}

\item If $z_0 > 0$, $\gamma_z \subset \text{exterior}(\gamma_w)$

\item If $z_0 < 0$, $\gamma_z \subset \text{interior}(\gamma_w)$
\end{itemize}
for if we can show this, then the rest of the argument is as for the liquid region, and the integration contours we end up with for the limit (from the residue at $z = w$) match up with those of $K^{z_0}$. See Figure \ref{fig:frozen_regions_gz} for an example.

Arguments similar to those given for the liquid region can be given to show that for each one of the $5$ intervals, one of the choices of critical point $z_0$ in the interval will give level curves as shown in the Figure, and therefore there exist contours $\gamma_z, \gamma_w$ satisfying the desired properties. We give an argument for the case $z_0 \in (\alpha_2, \infty)$ (shown in Figure \ref{fig:F3}) and omit the similar arguments in the other cases. In this case two critical points occur in $(\alpha_2, \infty)$: There will be a local maximum and a local minimum of $S(z)$ as $z$ moves from $\alpha_2$ to $\infty$ along the real axis. We let $z_0$ be the local maximum. There are two loops emerging from $z_0$, and one must be contained inside the other. This is because neither can intersect the real axis at a point $r > z_0$, since $S(z)$ is analytic. Thus, they intersect the real axis at points $x_1 < x_2$ with $x_2 < z_0$. However, by the behavior of $S$ on $(\alpha_2, \infty)$, the entire level set $\text{Re}(S(z)) = \text{Re}(S(z_0))$ does contain a point $x_3 > z_0$. If $x_1 < -\frac{1}{\beta}$ we have a contradiction because there must be some curve connecting $x_3$ with a different zero of $\text{Re}(S(z)-S(z_c))$ in the interval $(-\frac{1}{\beta}, \alpha_2)$, and two level curves cannot cross each other. On the other hand, if $x_1 > 0$, then it is impossible to account for the fact that there are either $2$ or $0$ points $z \in (\alpha_1, \alpha_2)$ with $\text{Re}(S(z)) = \text{Re}(S(z_0))$ without being led to a similar contradiction. Thus, we must have $x_1 \in (-\frac{1}{\beta}, 0)$, and to account for $x_3$, there must be a large level curve connecting a point $x_0 < -\frac{1}{\beta}$ with $x_3$ (this is not shown in the Figure). Thus, we may draw the contours as shown in Figure \ref{fig:F3}. This concludes the proof.

\end{proof}

Now we state a corollary which is a precise restatement of Theorem \ref{thm:mainthm}, in terms of convergence of correlation functions. It follows from the theorem above that the arctic curve $C$
is given by the locus of points $(x, u)$ such that $\partial_z S(z; x, u) = 0$ has a double root ($S$ is defined in \eqref{eqn:S}). Denote by $z_0 = z_0(x,u)$ the complex function defined in Theorem \ref{thm:loc} above. Explicit calculations using the correlation kernel $K^{z_0}$ (see Section \ref{sec:burgers} for a more detailed discussion) imply that 
$$s = \frac{1}{\pi}\big( \arg(1 + \beta z_0) + \arg(z_0) - \arg(z_0-1) - \arg(z_0-\alpha)\big), \qquad t = \frac{1}{\pi} \arg(z_0)$$
where $(s, t)$ are the slopes of the Gibbs measure with complex coordinate $z_0$, with respect to the fundamental domain displayed in Figure \ref{fig:mag}. By the theorem above, these are also the slopes of the limiting height function at $(x, u)$. Thus we obtain the following corollary.

\begin{cor}
Let $(x, u) \in \mathfrak{T} \setminus C$, and let $(s, t)$ be the slopes of the limiting height function at $(x, u)$. For any finite set of edges $(\mathbf{b}(X_i, U_i), \mathbf{w}(X_i', U_i'))$ of the square hexagon lattice, denoting by $e_i^N$ the edge $(\mathbf{b}(\floor{3 N x} + X_i, \floor{N u} + U_i), \mathbf{w}(\floor{3 N x} + X_i', \floor{N u} + U_i'))$ of the size $N$ tower graph, we have as $N \rightarrow \infty$ 
\begin{align*}
&P_N(\{\mathrm{perfect\; matchings\; } M : e^N_i \in M \; \forall i = 1,\dots, m \}) \\
&\rightarrow \pi_{s, t}(\{ \mathrm{perfect\; matchings\; } M : (\mathbf{b}(X_i, U_i), \mathbf{w}(X_i', U_i')) \in M \; \forall i = 1,\dots, m\}) .
\end{align*}
\label{cor:main}
\end{cor}
\begin{proof}
This follows from Theorem \ref{thm:loc}, together with the determinantal form \eqref{eqn:detcor} of the edge inclusion probabilities both for finite $N$ tower graphs and for the translation invariant Gibbs measures.
\end{proof}

\begin{remark}
Given the analysis of the double contour integral formula for $K$ in the previous proof with both of $\alpha_1, \alpha_2$ present, the proof of Theorem \ref{thm:mainthm} with general $\alpha_1, \alpha_2$ weights could be obtained immediately after making the appropriate changes to the definition of $K$ and of the family of Gibbs measure kernels $\{K^{z_0}\}$.
\end{remark}

\section{Growth Process: Current and Hydrodynamics}
\label{sec:growth}
\subsection{Shuffling: Spider Move Formulation}
\label{subsec:shuffling}
It is well known that the shuffling algorithm for domino tilings of the \emph{Aztec diamond} comes from a resampling procedure induced by the \emph{spider move}, which is a local transformation of the graph and its edge weights. Roughly speaking, doing this move at the deterministic level preserves partition functions, and as a result the naturally associated Markov mapping preserves probability measures. We briefly review the essential facts, and refer the reader to \cite{propp2003generalized, ZhangDominoHydrodynamics} and references therein for more details.

\begin{figure}
         \centering
         \includegraphics[scale=0.8]{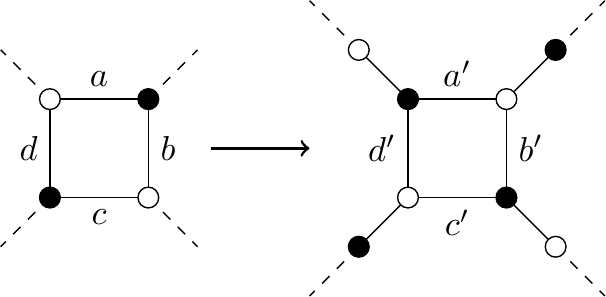}
\caption{The weights update as $a' = \frac{c}{\Delta}, c' = \frac{a}{\Delta}, d' = \frac{b}{\Delta}, b' = \frac{d}{\Delta}$, with $\Delta = a c + b d$. The relationship between the local partition functions $Z$ and $Z'$ of the two patches shown above is the following: Subject to any boundary condition (a boundary condition is a subset of the boundary vertices which are occupied by edges outside of the patch) $Z = \Delta Z'$.}
\label{fig:spider}
\end{figure}

More precisely, if we have a finite graph $G$ embedded in the plane or the torus with a degree $4$ face $F$ with edge weights $a,b,c,d$ as shown in Figure \ref{fig:spider}, then we can replace it with the graph $G'$ where the patch around the face $F$ is altered as on the right of Figure \ref{fig:spider}, and the corresponding dimer partition functions $Z_G, Z_{G'}$ are related by 
$$Z_G = \Delta Z_{G'} \;\;.$$

We can use the spider move at $F$ to define a Markov mapping, which takes a random configuration $M$ sampled from the Boltzmann measure $\mu_G$ to a random configuration $M'$ sampled from the Boltzmann measure $\mu_{G'}$. The definition is as follows: Given the matching $M$, we set the matching $M'$ have the same subset edges as $M$ outside of the patch around $F$. Then, this set of edges determines a set of partial matchings on the patch shown on the right of Figure \ref{fig:spider} which can be used to complete $M'$ to a perfect matching. We resample the edges of the partial matching according to the Boltzmann measure on that patch with the new weights. By an abuse of terminology, we will refer to such a Markov move as a \emph{spider move} or \emph{urban renewal} at the face $F$.

\begin{figure}
\begin{subfigure}[b]{0.3\textwidth}
         \centering
         \includegraphics[scale=0.30]{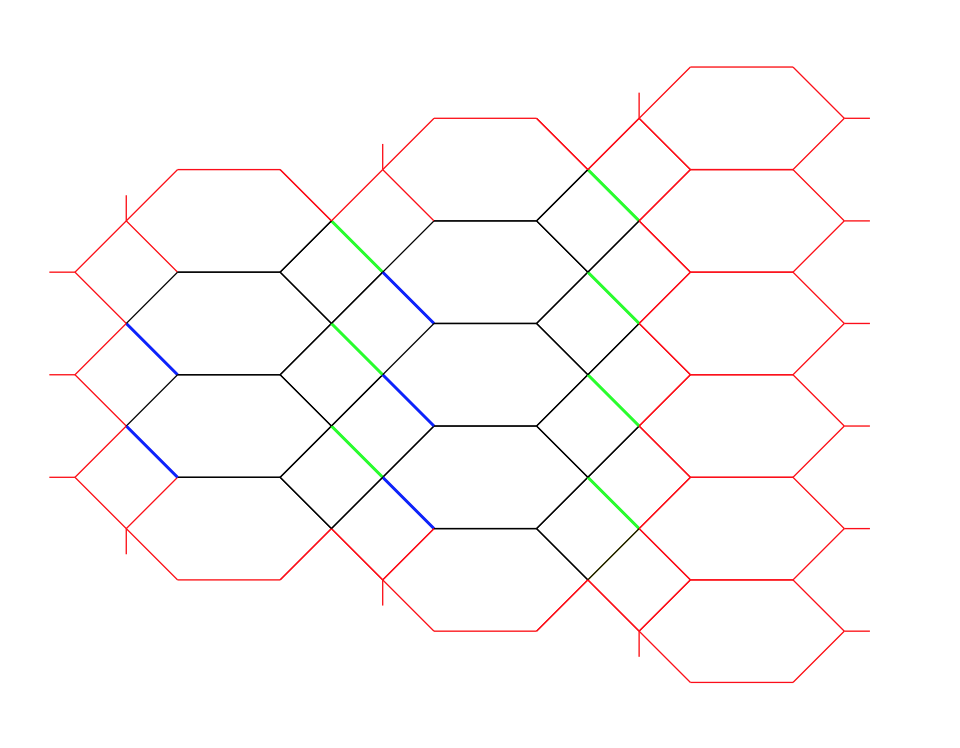}
         \caption{Decoration of the boundary of the tower graph of size $N = 2$. Weights on red edges are not shown, but should follow the same pattern.}
         \label{fig:shs0}
     \end{subfigure}
     \hspace{40 pt}
     \begin{subfigure}[b]{0.3\textwidth}
         \centering
         \includegraphics[scale=0.25]{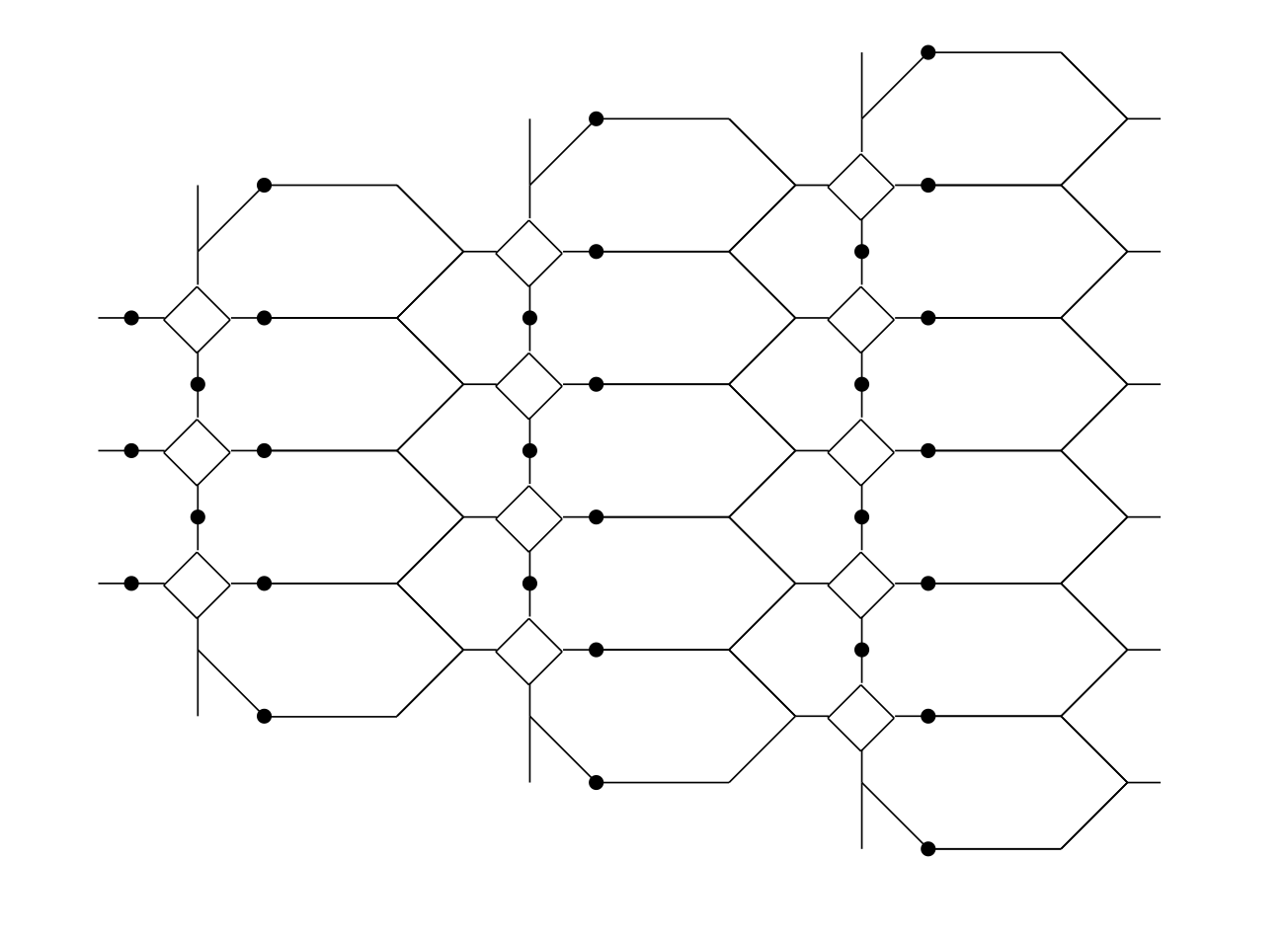}
          \vspace{-20 pt}
         \caption{The first round of spider moves. The black dots indicate the $2$-valent vertices at which we will perform contraction.}
         \label{fig:shs1}
     \end{subfigure} \\
     \begin{subfigure}[b]{0.3\textwidth}
         \centering
         \includegraphics[scale=0.30]{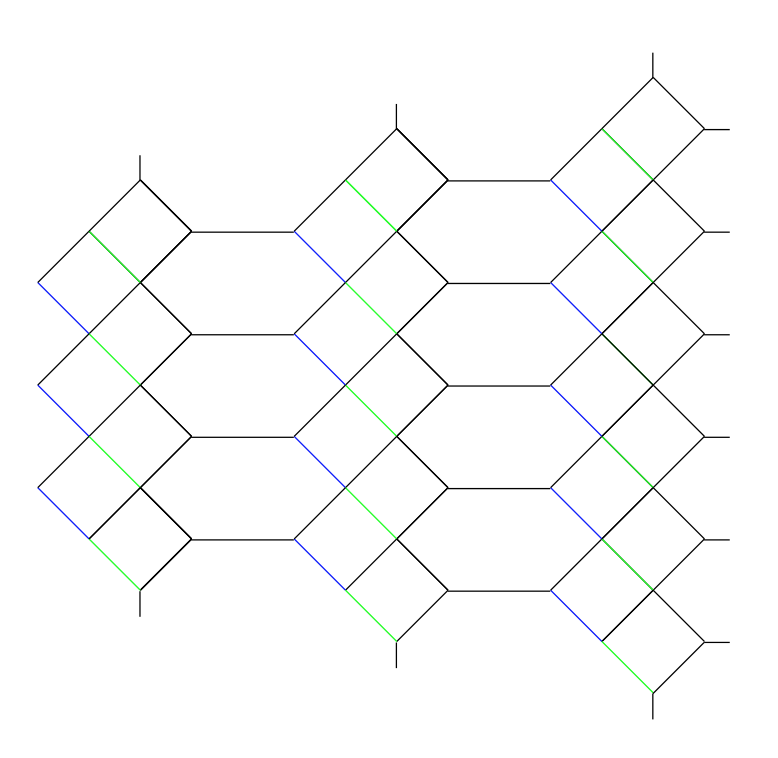}
         \caption{A gauge equivalent graph after the first round of spider moves.}
         \label{fig:shs2}
     \end{subfigure}
     \hspace{40 pt}
     \begin{subfigure}[b]{0.3\textwidth}
         \centering
         \includegraphics[scale=0.30]{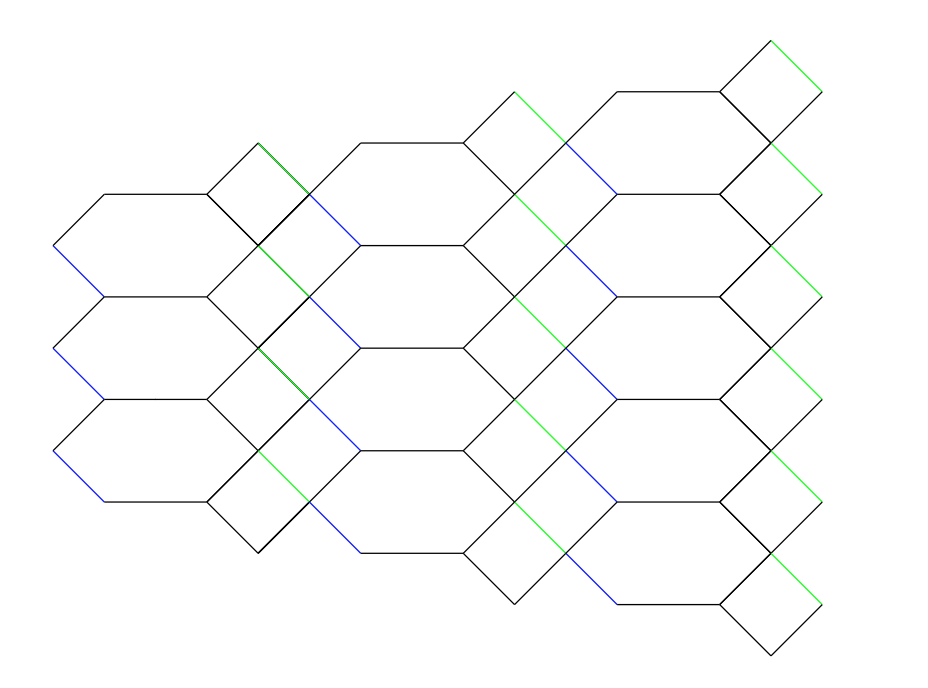}
         \caption{A gauge equivalent graph after the second round of spider moves (the second round is not explicitly shown). }
         \label{fig:shs3}
     \end{subfigure}
\caption{The sequence above illustrates the sequence of spider moves which, starting with a tower graph of size $N$, produces a tower of size $N+1$. The green edges have weight $\beta$, blue ones have weight $\alpha$, and all other edges have weight $1$. Randomizing each spider move gives the corresponding shuffling algorithm which samples from the Boltzmann measure on the size $N + 1$ tower given a sample on a tower of size $N$.}
\label{fig:shuffle}
\end{figure}

 Now we describe the sequence of spider moves to perform in order to obtain a tower of size $N + 1$ from one of size $N$. For any subgraph of the square hexagon graph, define \emph{faces of type 1} to be those which are at $X$ coordinate one smaller than a hexagonal face. We record the following fact as a proposition, which follows trivially from the facts stated above, and some elementary combinatorial reasoning.
 
 \begin{prop}
 Starting with a tower of size $N$, we decorate the graph as shown in Figure \ref{fig:shs0}. Matchings of this graph are in weight-preserving correspondence with matchings of the original tower graph. The following two rounds of spider moves map a random matching of a size $N$ tower to one of a size $N + 1$ tower.
\begin{enumerate}[1.]
\item Perform the spider move at each face of type $1$. 
\item Contract the two-valent black vertices, as shown in Figure \ref{fig:shs1}, and then perform the spider move at each face which is now of type $1$. Then we again contract two-valent black vertices.
\end{enumerate}
Furthermore, this update step is equivalent to the dynamics on interlacing arrays described in Section \ref{subsec:res}.
\label{prop:shuffling}
 \end{prop}

In Figure \ref{fig:shuffle} we illustrate this procedure for $N = 2$ (omitting the spider moves of the second step). Thus, using the corresponding Markov mapping for each spider move, to generate a random matching on a tower of size $N$ we may perform $N$ steps of the randomized shuffling.

\subsection{Stationary Dynamics and Current} 
Consider a point with continuum coordinates $(x, u)$ in the bulk of the size $N$ tower graph, for $N$ large. The local statistics in a finite neighborhood $\mathcal{B}$ of the point are approximated by an ergodic, translation invariant Gibbs measure. As a result, the distribution of the tiling at time $N+1$ inside of $\mathcal{B}$ is approximately given by: 1. Sampling a dimer cover from a translation invariant Gibbs measure, and 2. running one step of a corresponding full plane Markov chain on $G$. At leading order, one step of the Markov chain does not change the slope of the macroscopic height function at $(x, u)$, so the Gibbs measure we see in $\mathcal{B}$ at time $N+1$ should be the same, so we expect the Gibbs measures to be stationary under the full plane Markov chain.

In this section we define the full plane dynamics, which act on full plane dimer configurations. We will define a Markov process on dimer configurations of the square hexagon lattice $G$ such that for any $(s, t)$ the ergodic, translation invariant Gibbs measure with slopes $(s, t) $ is stationary. Recall that a face of a square hexagon graph is called a \emph{face of type 1} if the adjacent faces to the east are hexagonal faces.

\begin{figure}
  \includegraphics[scale=0.4]{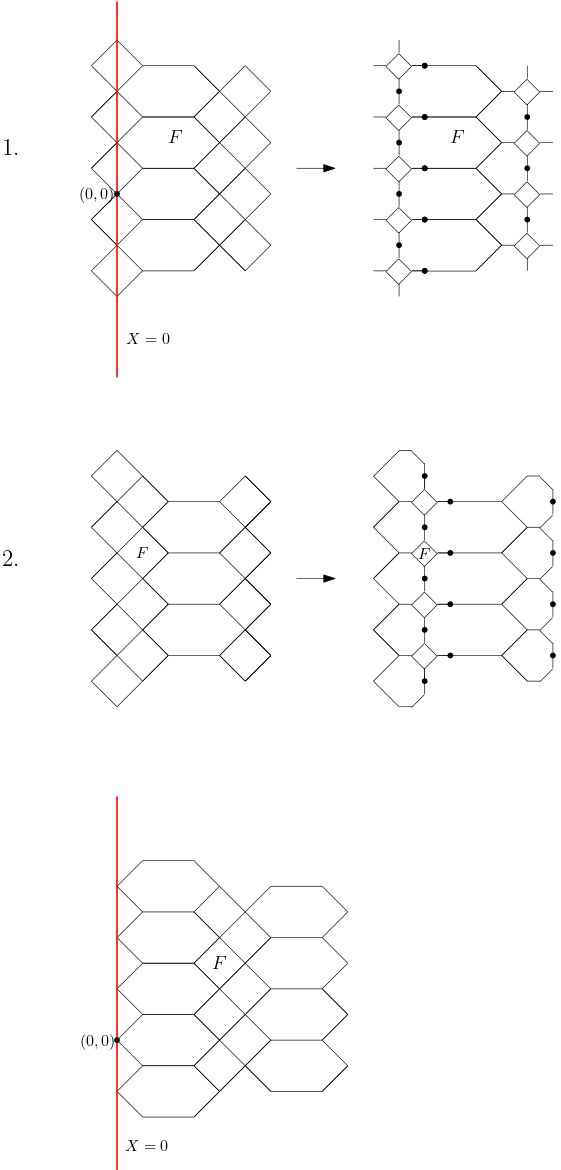}
\caption{Shown above are the two update steps involved in one time step of the full plane Markov chain. We keep track of the face $F$ as we do the shuffling. At time $0$, $F$ has coordinates $(0, 0)$, and at time $1$ it has coordinates $(1, 0)$.}
\label{fig:fps}
\end{figure}

\begin{defn}[Full plane dynamics]
 Suppose our initial matching of $G$ is $M_0$. One step of the Markov chain is defined by the following two steps:
\begin{enumerate}[1.]
\item  Perform a spider move at the faces of type 1, and contract $2$-valent vertices. This gives a random matching $M_{\frac{1}{2}}$ of a graph $G' \cong G$.
\item Perform a spider move at the faces of $G'$ which are of type $1$, and contract $2$-valent vertices. This gives a random matching $M_{1}$ of a graph $G'' \cong G$. 
 \end{enumerate}
 \label{def:FP}
\end{defn}
Even though $G''$ is isomorphic to $G$, in order to specify the exact identification of $G''$ with the square-hexagon lattice $G$ we specify the coordinates we use to index the vertices and faces. The coordinate axes will be parallel to the same $(X, U)$ coordinate axes, so we must specify the position of the origin. After any sequence of spider moves on a graph, the faces of the new graph are combinatorially in a natural correspondence with the faces of the old graph. Using this fact, we choose the identification of $G''$ with $G$ as follows:

\begin{defn}[Coordinates]
The coordinates on faces of $G''$ (see definition \ref{def:FP}) are determined by the following requirement: If we have a face $F$ of $G''$ which corresponds combinatorially to the face of $G$ at position $(X, U)$ before the spider moves, its coordinates on $G''$ are $(X+1, U)$. See Figure \ref{fig:fps}.
\end{defn}

Note that we must additionally specify how to update the height function, as a dimer configuration corresponds to a height difference profile, rather than a height function itself. For this it is sufficient to specify the new height function at a single face. We choose this update rule so that the height function of the full plane Markov chain corresponds to the height evolution under the shuffling algorithm on the finite tower graph. If $H_0$ is the height function of the initial matching and $H_1$ is the height function after one step of the chain, it is not hard to see that for any two faces $F, F'$ at which a spider move is not performed
$$H_1(F') - H_1(F) = H_0(F') - H_0(F) \;\;.$$
See, e.g. Lemma 2.3 of \cite{ZhangDominoHydrodynamics}. Thus, it makes sense to define $H_1(F) = H_0(F)$ for all $F$ at which no spider move happened. Using this fact, and noting that the face at position $(X, U)$ at time $1$ corresponds combinatorially to the face at $(X-1, U)$ at time $0$ (see Figure \ref{fig:fps}), we make the following definition.

\begin{defn}[Height function update]
 Suppose our initial matching of $G$ is $M_0$, with initial height function $H_0$. The height function of the new matching $M_1$ is fully determined by setting 
 $$H_1(X, U) = H_{0}(X-1, U) $$
 at faces $(X, U)$ such that we did not do a spider move.
\end{defn}

We denote by $H_k$ the height function after $k$ steps of the above dynamics. Now we compute the average speed of growth of $H_k$ as $k$ increases, assuming $H_0$ was sampled from an ergodic Gibbs measure. The following lemma is needed to make this precise.

\begin{lemma}
The full plane dynamics preserves the Gibbs measure $\pi_{s, t}$ for any allowed slopes $(s, t)$.
\end{lemma}

\begin{proof}
First we note that there exists a natural analog of the shuffling algorithm described in Definition \ref{def:FP}, which is defined on the set of dimer configurations on the $N \times N $ torus graph $G / (N \mathbb{Z})^2$ constrained to have fixed height changes $\floor{N s}, \floor{N t}$ in the horizontal and vertical directions. Furthermore, we note that this chain preserves the Gibbs measure $\pi^N_{s, t}$ on the same set of dimer configurations.

Next, for any $L, \epsilon > 0$, for all $N$ large enough there exists a coupling of dimer configurations $M \sim \pi_{s, t}$ on $G$ and $M_N \sim \pi^N_{s, t}$ on $G / (N \mathbb{Z}^2)$ such that with probability $\geq 1- \epsilon$ the configurations agree on all edges around faces inside of a radius $L$ ball $\mathcal{B}_L$ around $(X, U) = (0, 0)$ (where we interpret this ball modulo $N \mathbb{Z}^2$ in the torus situation). The configuration inside $\mathcal{B}_{L-5}$ in both graphs after one step of shuffling only depends on the configuration in $\mathcal{B}_L$, as well as the outcomes of the Bernoulli random variables used for creations at faces inside $\mathcal{B}_L$. Let $M(1), M_N(1)$ denote the two dimer configurations after one step of the shuffling algorithm on each graph, which are coupled by using the same Bernoulli random variables for creation steps at faces inside of $\mathcal{B}_L$ for which both matchings have a creation (and creations are done independently otherwise). With this construction, we obtain a coupling under which the configurations $M(1)$ and $M_N(1)$ agree inside of $\mathcal{B}_{L-5}$ at time $1$ agree with probability $\geq 1- \epsilon$. However, $M_N(1) \stackrel{d}{=} M_N$ by the stationarity on the torus, so this coupling implies the lemma.
\end{proof}

We also define an important quantity known as the \emph{current}, which is a function $J(s, t)$ of the allowed pairs of slopes $(s, t)$. One can easily check that the definition below is independent of the choice of face $(X, U)$.
\begin{defn}
Fix $(X, U) \in \mathbb{Z}^2$. Define the \emph{current} $J(s, t)$ by the quantity
\begin{equation}
\label{eqn:curr}
J(s, t) \defeq \mathbb{E}_{\pi_{s, t}}[ H_1(X, U) - H_0(X, U) ] \;\;.
\end{equation}
\end{defn}

Since the set of Gibbs measures we consider are parameterized by a complex coordinate $z_0 \in \mathbb{H}$, we can write the current as a function of $z_0$, and (abusing notation) we write this as $J(z_0)$.

\begin{theorem}
The current is given by 
\begin{equation}
\label{eqn:curr}
J(z_0) = -\frac{1}{\pi} \arg(1 + \beta z_0) \;\;.
\end{equation}
\end{theorem}

\begin{proof}
Going through the definitions, it becomes clear that the expression in the equation above indeed does not depend on $f$ and is given by a single edge probability:
$$J(s, t) =- \pi_{s, t}\left( e \in M\right) \;\;.$$

Once we know this, it simply suffices to compute the edge probability using the corresponding entry of the inverse Kasteleyn matrix $K^{z_0}$:

\begin{align*}
J(z_0)&= - \mathbf{K}(\mathbf{b}(-1, 1), \mathbf{w}(0, 0))  K^{z_0}(\mathbf{w}(0, 0), \mathbf{b}(-1, 1)) \\
&= - \beta \frac{1}{2 \pi i  } \int_{\overline{z_0}}^{z_0} (1+ \beta z)^{-1}  d z  \\
&= - \frac{1}{\pi} \arg(1 + \beta z_0) \;\;.
\end{align*}

\end{proof}

\begin{remark}
\label{rmk:hydro}
We expect that for any fixed initial matching on $G$, after rescaling we should get the convergence of normalized random height functions
$$\frac{1}{N} H_{\floor{\tau N}}(\floor{x N}, \floor{u N}) \rightarrow h_\tau(x, u)$$ 
where the deterministic height function $h_\tau(x, u)$ satisfies the PDE: 
\begin{equation}
\label{eqn:hl}
\partial_\tau h_\tau = J(\nabla h_\tau) \;\;.
\end{equation}
This statement is analogous to the theorems presented in \cite{ZhangDominoHydrodynamics} and \cite{laslier2017lozenge} in the case of the square and hexagonal lattices. In particular, the hydrodynamic limit $h_\tau$ of the height function under the evolution of the shuffling algorithm for tower graphs should satisfy this PDE in the hydrodynamic limit. It is shown in \cite{borodin2018two} that if $h_\tau$ evolves according to equation \eqref{eqn:hl}, and if $h_\tau$ also satisfies the Euler Lagrange equations of the limit shape for each $\tau$, then as a function of the complex coordinate, $J(z_0)$ must be harmonic. We indeed see that this is the case. In other words, the evolution of $h_\tau$ is consistent with the Burgers equation satisfied by $z_0$ at each time $\tau$. 
\end{remark}

Using explicit formulas for the limit shape for $\alpha = 1$, we verify in Section \ref{sec:burgers} that $h_\tau$ satisfies the hydrodynamic limit equation \eqref{eqn:hl} for the initial condition corresponding to tower graphs. Furthermore, one may obtain a formula for the current $J(s, t)$ in terms of local slopes by inverting the mapping from $z_0$ to slopes $(s, t)$. This map can be described by the elementary geometric picture in Figure \ref{fig:angles}. We elaborate more on these two points in the next section.

\section{Limit Shapes and Complex Coordinate}
\label{sec:burgers}
\subsection{Bijection Between Complex Coordinate $z_0$ and Slopes}
\label{subsec:z}

\begin{figure}
\includegraphics[scale=0.8]{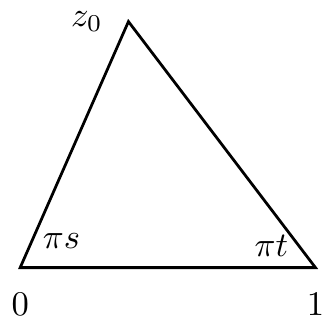}
\caption{The relationship between the complex coordinate $z_0$ and local slopes $(s, t)$ in the case of lozenge tilings.}
\label{fig:angles_lozenge}
\end{figure}

\begin{figure}
\includegraphics[scale=0.8]{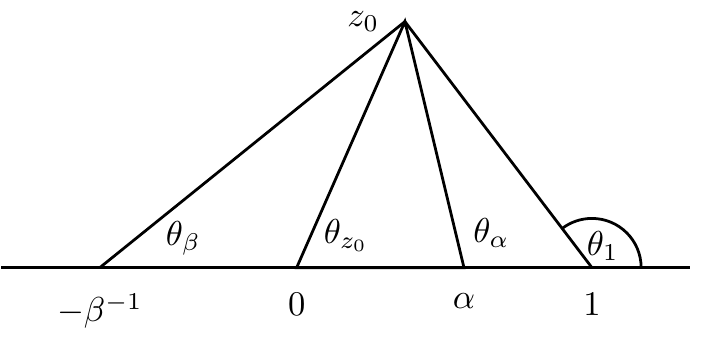}
\caption{From $z_0$ we compute the edge inclusion probabilities as a function of the angles above using equations \eqref{eqn:p1}-\eqref{eqn:p4}.}
\label{fig:angles}
\end{figure}

In this section we make explicit the correspondence between the complex coordinate $z_0$ defining a Gibbs measure and the slopes $(s, t)$ corresponding to that Gibbs measure. According to the general theory of Kenyon-Okounkov-Sheffield \cite{KOS2006}, pairs of complex numbers $(z_0, w_0)$ such that $z_0 \in \text{int}(\mathbb{H})$ and $P(z_0, w_0) = 0$ are in one-to-one correspondence with pairs of liquid phase slopes via the relation
\begin{equation}
(s, t) =  \pm \frac{1}{\pi} (-\arg w_0, \arg z_0) \text{ (mod $\mathbb{Z}^2$)} \;\;.
\label{eqn:args}
\end{equation}
For the dimer model on several other lattices, this correspondence can be made explicit with a simple geometric picture. For example, in the case of lozenge tilings, the pair $(z_0, w_0)$ is determined by a choice of $z_0 \in \mathbb{H}$, and one may simply draw the triangle of Figure \ref{fig:angles_lozenge} to see the correspondence between $z_0$ and the local slopes.

The mapping has a similar description in our case; we may use the angles of Figure \ref{fig:angles} to compute $(s, t)$ from $z_0$. Indeed, we use the form of the inverse Kasteleyn of lemma \ref{lem:gibbsmeas} to compute the following edge probabilities.
\begin{align}
\pi_{z_0}(( \mathbf{w}(0, 0), \mathbf{b}(-1, 1)) \in M) &= 
\frac{1}{\pi} \arg(1 + \beta z_0) = \frac{1}{\pi} \theta_{\beta} \label{eqn:p1}\\
\pi_{z_0}((\mathbf{w}(-2, 1), \mathbf{b}(-1, 0)) \in M) &= \frac{1}{\pi} (\arg(z_0 - 1) - \arg(z_0)) = \frac{1}{\pi}(\theta_1 - \theta_{z_0})
 \\
\pi_{z_0}(( \mathbf{w}(0, 0), \mathbf{b}(1, 0)) \in M) &= 1- \frac{1}{\pi}  \arg(z_0-\alpha) = \frac{1}{\pi}(\pi - \theta_{\alpha})
\\
\pi_{z_0}((\mathbf{w}(1, 0), \mathbf{b}(1, 0)) \in M) &= \frac{1}{\pi} \arg z_0 = \frac{1}{\pi} \theta_{z_0} \label{eqn:p4} \;\;.
\end{align}
All arguments above are chosen in $[0, \pi]$. In particular, using the definition of the height function and our choice of fundamental domain (Figure \ref{fig:mag}), the slopes of the Gibbs measure are given by
\begin{align*}
s &= \frac{1}{\pi} ( \theta_\beta + \theta_{z_0} -\theta_1  -\theta_\alpha)  \\
t &= \frac{1}{\pi} \theta_{z_0} \;\;.
\end{align*}
Note that since we can write $w_0 =  \frac{(z_0- 1)(z_0 - \alpha)}{z_0 (1 + \beta z_0)}$, the above is consistent with equation \eqref{eqn:args}.

Using law of sines, one observes (letting $r_1 = |z_0|$)
\begin{align*}
\frac{\sin(\theta_\beta) }{\sin(\theta_{z_0} - \theta_\beta)} &= \beta r_1 \\
\frac{\sin(\theta_\alpha) }{\sin(\theta_{\alpha} - \theta_{z_0})} &= \frac{r_1}{\alpha}  \\
\frac{\sin(\theta_1) }{\sin(\theta_{1} - \theta_{z_0})} &= r_1
\end{align*}
and thus
\begin{align*}
\frac{\sin(\theta_\beta) }{\sin(\theta_{z_0} - \theta_\beta)} &= \alpha \beta \frac{\sin(\theta_\alpha) }{\sin(\theta_{\alpha} - \theta_{z_0})} \\
\frac{\sin(\theta_1) }{\sin(\theta_{1} - \theta_{z_0})}  &= \alpha \frac{\sin(\theta_\alpha) }{\sin(\theta_{\alpha} - \theta_{z_0})} \;\;.
\end{align*}
From the above equations it becomes clear that there are only two independent parameters. Using these equations, it is possible to express all of the angles $\theta_{z_0}, \theta_1, \theta_{\alpha}, \theta_{\beta}$ in terms of $s, t$, and thus obtain $z_0$ as a function of $(s, t)$.

\subsection{Identifying Critical $z_0(x, u)$ with Complex Coordinate of Kenyon-Okounkov}

Now we will identify the complex coordinate $z_0(x, u)$ which arises naturally in the saddle point analysis of double contour integrals with the complex coordinate obtained from solving the Burgers equation.

It follows from Theorem 1 of \cite{OkounkovKenyon2007Limit}, that if one expresses the Euler-Lagrange equation for the limiting height function $h(x, y) = \lim_{N \rightarrow \infty} \frac{1}{N} H_N(\floor{N x}, \floor{N y})$ in terms of the corresponding complex coordinates $(z, w)$, which satisfy $\nabla h = \frac{1}{\pi} (-\arg w, \arg z) $ and $P(z, w) = 0$, then the equation for the limit shape is 
$$
\frac{z_x}{z} + \frac{w_y}{w} = 0 \;\;.
$$
The $(x, y)$ coordinates above are induced by the choice of fundamental domain (see Figure \ref{fig:z2lattice}), and are related to $(x, u)$ by $y = x + u$. One approach is to explicitly check that $z_0$ satisfies the same Burgers equation and the correct boundary conditions after changing coordinates. Alternatively, we can use the \emph{method of complex characteristics}, see Corollary 1 of \cite{OkounkovKenyon2007Limit}, to find a solution to the Burgers equation matching $z_0$. We take this approach.

The method of characteristics works as follows: One can simply check that if $Q(z, w)$ is an analytic function, the solutions $(z(x, y), w(x,y))$ of
\begin{equation}
\label{eqn:char}
\begin{cases}
Q(z, w) = x z P_z + y w P_w \\
P(z,w ) = 0
\end{cases}
\end{equation}
are solutions to the complex Burgers equation. The analytic function $Q$ encodes the boundary conditions for the height function. If we choose $Q(z, w) = - \beta w$, then solving equations \eqref{eqn:char} for $z(x,y)$, we get that $z(x, x+u) = z_0(x, u)$. Indeed, the system \eqref{eqn:char} leads to an equation $(\text{cubic polynomial in $z$}) = 0$, and the cubic polynomial is the exact same as the one appearing in the equation $\partial_z S(z; x, u) = 0$ for the critical point. This polynomial equation for $z$ is 
\begin{align*}
&-\alpha u - 2 \alpha x + (\alpha \beta + u + \alpha u - \alpha \beta u + x + \alpha x - 3 \alpha \beta x) z 
 \\  
 &+ (-\beta - 
    \alpha \beta - u + \beta u + \alpha \beta u + 2 \beta x + 2 \alpha \beta x) z^2 + (\beta - \beta u - \beta x) z^3 = 0 \;\;.
\end{align*}
So we see that indeed the complex coordinate coming from the saddle point analysis is exactly the solution of the Burgers equation.

\begin{example}[Uniform weights]
\label{ex:uni}
If $ \alpha = 1$ the cubic polynomial above factors as 
$$(1- z) (-u - 2 x + (\beta + u - \beta u - 3 \beta x) z + (-\beta + \beta u + \beta x) z^2) $$
and further specializing $\beta = 1$ (this is the case of uniformly random perfect matchings) and solving, we get 
$$z(x, u) = \frac{-1 + 3 x + \sqrt{(1 - 3 x)^2 - 
  4 (-u - 2 x) (-1 + u + x)}}{2 (-1 + u + x)} \;\;.$$

Note that the complex coordinate for $h_\tau(x, u)$ in the time varying situation must be $z(x/\tau, u/\tau)$ by the self-similarity of the domain as $\tau$ varies. Using this and the explicit formula above, one can easily check the consistency with equation \eqref{eqn:hl}. Indeed, we have 
$$h_\tau(x, u) = - \frac{1}{\pi} \int_{u}^\infty \arg(z(x/\tau, u'/\tau)) \; du' \;\;.$$
In fact we can compute 
$$\partial_\tau \left(- \int_{u}^\infty \log(z(x/\tau, u'/\tau)) \; du' \right) = \frac{1}{\tau^2} \int_u^\infty \frac{1}{z(x/\tau, u'/\tau)} (z_x x + z_{u'} u') \; du \;\;.$$
So replacing $J(z)$ in equation \eqref{eqn:hl} with $-\frac{1}{\pi} \log(1 + z)$, and taking the $u$ derivative of both sides (and ignoring the prefactor of $\pi$), it suffices to check the identity 
\begin{align*}
\frac{(z_x x/\tau + z_{u} u/\tau)}{z} = \frac{1}{1 + z} z_u
\end{align*}
for $z, z_x, z_u$ all evaluated at $z(x/\tau, u/\tau)$, which can be readily checked using the explicit formula.

\end{example}

\section{Appendix}
\subsection{Isoradial embeddings}
\begin{figure}[!ht]
    \centering 
\begin{subfigure}{1.0\linewidth}
\centering
  \includegraphics[scale=0.4]{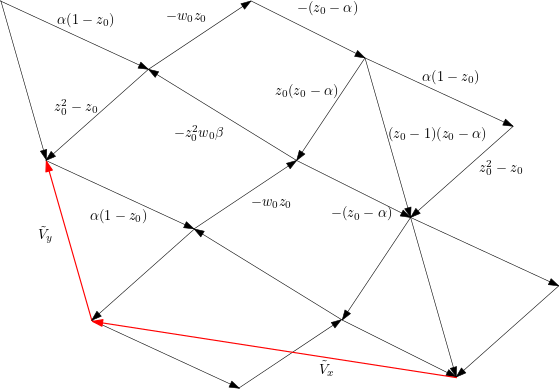}
  \caption{}
  \label{fig:i1}
\end{subfigure} \\ \vfill
\begin{subfigure}{1.0\linewidth}
\centering
  \includegraphics[scale=0.4]{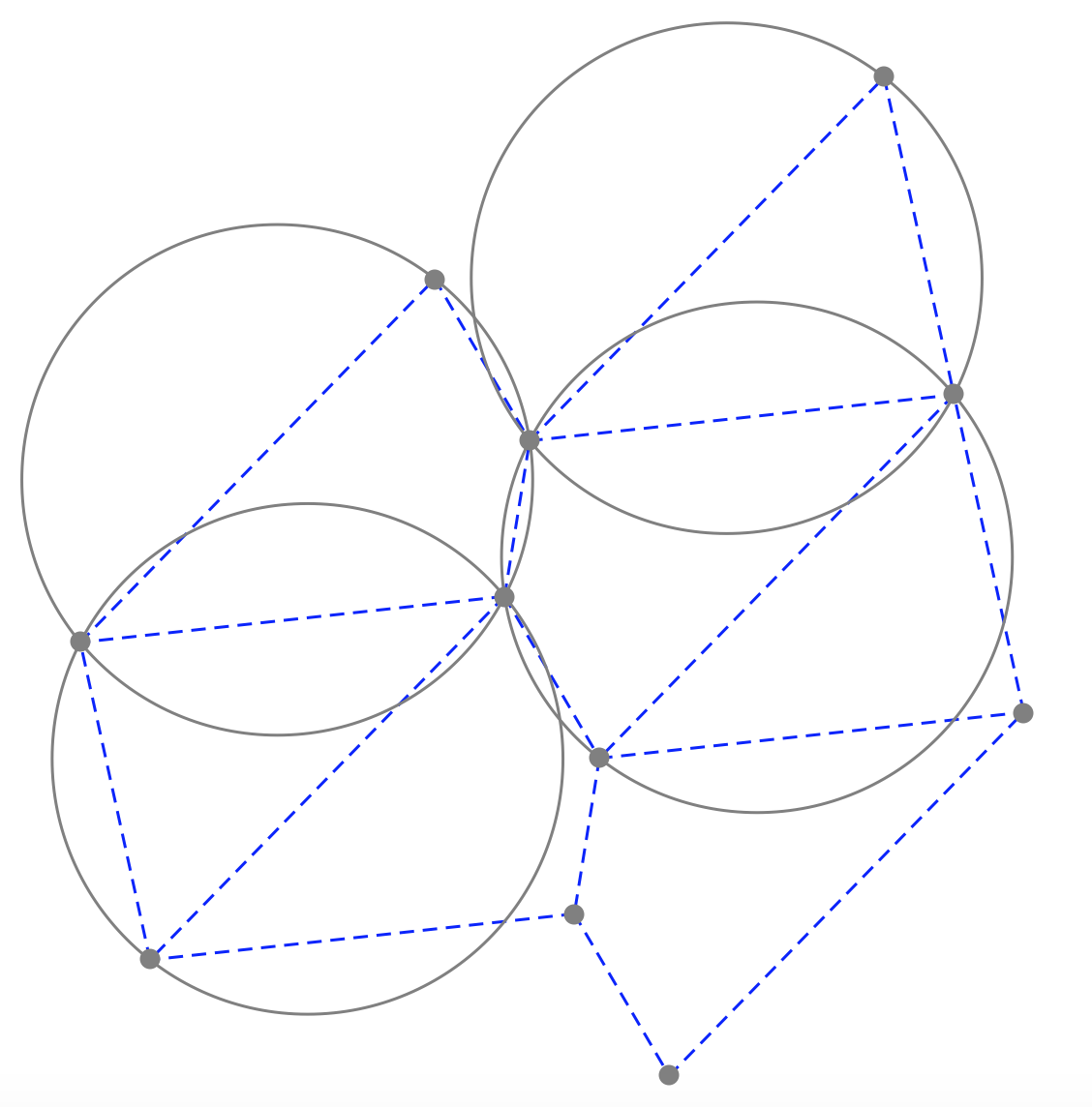}
  \caption{}
  \label{fig:i2}
\end{subfigure} 
\caption{(A): We present a schematic of an isoradial embedding of $G^*$ (the red edge $\tilde{V}_x$ is not part of the embedding), for which the inverse Kasteleyn $\overline{\partial}^{-1}$ defined in Section 4 of \cite{K_lapDet_2000} gives the Gibbs measure corresponding to $z_0$. We use the notation $w_0 = w(z_0)$, so $P(z_0, w_0) = 0$. The boundary edges of faces sum to $0$ as a consequence of the equation $P(z_0, w_0) = 0$. The embedded graph is invariant under translations by $\tilde{V}_x = w_0 z_0 + (z_0 - \alpha) +\alpha (z_0 - 1)$ and $\tilde{V}_y = -(z_0 - 1)(z_0 -\alpha)$ (compare with Figure \ref{fig:z2lattice}). (B): We show the isoradial embedding of $G^*$ with $\alpha = \beta = 1$ for $z_0= 0.9177956164184642 + 0.7575595655669651 i$. We only show circles around the four different types of faces in the graph (all other faces are translates of one of these four).}
\label{fig:isorad}
\end{figure}

Another way to obtain the slopes $(s, t)$ from $z_0$ using a geometric picture is to determine the \emph{isoradial embedding} of $G$, such that the naturally associated ergodic Gibbs measure is the one with slopes $(s, t)$. An isoradial embedding can be described as an embedding into $\mathbb{R}^2$ of the dual graph $G^*$ such that each face is inscribed in a circle of the same radius, and each isoradial embedding induces a natural \emph{critical weight function} on the edges of $G$ (see \cite{K_lapDet_2000} for detailed definitions). 

There is a natural inverse to the Kasteleyn matrix, which is defined by the embedding, and this inverse induces an ergodic Gibbs measure on dimer covers of $G$ \cite{de_Tili_re_2006}. In Figure \ref{fig:isorad} we present an isoradial embedding of $G$ such that this inverse induces the measure $\pi_{z_0}$. Given $z_0$ in the upper half plane, the embedding $f : \mathcal{L}^* \rightarrow \mathbb{C}$ of the dual of the square hexagon lattice can be described explicitly as follows. Denote by $w_0$ the unique complex number such that $P(z_0, w_0) = 0$. First, let $f(0,0) = 0$. Then, we prescribe values for $f(F') - f(F)$ for each edge $(F, F')$ of $\mathcal{L}^*$ dual to the edge $e$ (where the dual edge crosses $e$ with the white vertex on the right) as follows
\begin{itemize}
\item $f(F') - f(F) = -(z_0-1)(z_0-\alpha)$ if $e = \w(X, U) - \b(X, U) $ and $X = 1 \; (\text{mod } 3)$ 
\item $f(F') - f(F) =w_0 z_0$ if $e = \w(X, U) - \b(X-1, U) $ and $X = 0 \; (\text{mod } 3)$
\item $f(F') - f(F) = z_0^2-z_0$ if $e = \w(X, U) - \b(X+1, U) $ and $X = 0 \; (\text{mod } 3)$
\item $f(F') - f(F) = z_0 (z_0-\alpha)$ if $e = \w(X, U) - \b(X+1, U) $ and $X = 1 \; (\text{mod } 3)$
\item  $f(F') - f(F) = -z_0^2 w_0 \beta$ if $e = \w(X, U) - \b(X-1, U+1)$ and $X = 0 \; (\text{mod } 3)$
\item $f(F') - f(F) = \alpha (1- z_0)$ if $e = \w(X, U) - \b(X+1, U-1) $ and $X = 0 \; (\text{mod } 3)$
\item $f(F') - f(F) = -(z_0-\alpha)$ if $e = \w(X, U) - \b(X+1, U-1) $ and $X = 1 \; (\text{mod } 3)$
\end{itemize}

 That the edges around any face of $\mathcal{L}^*$ indeed sum to $0$ is a consequence of the equation $P(z_0, w_0) = 0$. Thus, we indeed have an immersion of the graph $\mathcal{L}^*$. That this is indeed an embedding, and furthermore that it defines an isoradial embedding, can be deduced from elementary geometric considerations.

\clearpage

\bibliographystyle{alpha}

\bibliography{bib}

\end{document}